\documentclass[12pt]{article}
\pdfoutput=1
\usepackage{amssymb}
\usepackage{amsmath}
\usepackage{amsthm}
\usepackage{lineno}
\usepackage{pdflscape}
\usepackage{array}
\usepackage{rotating}
\usepackage{soul}

\usepackage{graphicx}
\usepackage{verbatim}

\usepackage{url}

\newtheorem{proposition}{Proposition}[section]
\newtheorem{remark}{Remark}[section]
\newtheorem*{conjecture*}{Conjecture}

\title{A Shallow Water Model for Water Flow on Vegetated
  Hillslope}

\author{Stelian Ion$^*$, Dorin Marinescu$^*$, Stefan-Gicu
  Cruceanu\footnote{{Gheorghe Mihoc-Caius Iacob''
      Institute of Mathematical Statistics and Applied
      Mathematics, Romanian Academy, 050711 Bucharest,
      Romania, {\tt emails: ro\_diff@yahoo.com,
        marinescu.dorin@ismma.ro, gcruceanu@ismma.ro}.}}}

\date{}

\begin{document}
\maketitle

\begin{abstract}
  The hillslope hydrological processes are very important in
  watershed hydrology research.  In this paper we focus on
  the water flow over the soil surface with vegetation in a
  hydrographic basin.  We introduce a PDE model based on
  general principles of fluid mechanics where the unknowns
  are water depth and water velocity.  The influence of the
  plant cover to the water dynamics is given by porosity (a
  quantity related to the density of the vegetation), which
  is a function defined over the hydrological basin.  Using
  finite volume method for approximating the spatial
  derivatives, we build an ODE system which constitutes the
  base of the discrete model we will work with.  We discuss
  and investigate several physical relevant properties of
  this model.  Finally, we perform different quantitative
  validation tests by comparing numerical results with exact
  solutions or with laboratory measured data.  We also
  consider some qualitative validation tests by numerically
  simulating the flow on a theoretical vegetated soil and on
  a real hydrographic basin.
\end{abstract}

\section{Introduction}
The mathematical modeling of hydrodynamic processes in
hydrographic basins is a challenging subject.  There are two
main difficulties: the complexity of the phenomena and the
multitude of the factors involved in the process.  However,
there are plenty of performant models dedicated to some
specific aspects of the hydrodynamic processes only.  To
review the existent mathematical models is beyond the
purposes of this paper, but we can group them into two large
classes: physical base models and regression models.  The
most known regression models are the unit hydrograph
\cite{dooge} and universal soil loss \cite{rusle, wisch}.
From the first class, we mention here a few well known
models: SWAT \cite{swat}, SWAP \cite{swap} and KINEROS
\cite{kineros}.  Due to the complexity and heterogeneity of
the processes (see \cite{mcdonnel}), models in this class
are not purely physical because they need additional
empirical relations.  The main difference between models
here is given by the nature of the empirical relations.  For
example, in order to model the surface of the water flow,
SWAP and KINEROS use a mass balance equation and a closure
relation, while SWAT combines the mass balance equation with
the momentum balance equation.  A very special class of
models are cellular automata which combine microscale
physical laws with empirical closure relations in a specific
way to build up a macroscale model, e.g. CAESAR
\cite{caesar, sds-ose}.  In this paper, we introduce a model
that extends the shallow water equations and takes into
account the presence of vegetation on the soil surface.
Such kind of extension is not new; we remind here some
applications that use similar models: flow through rigid
vegetation \cite{minh}, flash food propagation in the urban
area \cite{GUINOT201240, GUINOT2017133} and tsunami advance
in the coastal region \cite{TANG201321}.  Our model results
as an asymptotic limit of a space averaged version of the
mass conservation equation and linear momentum balance
equations \cite{imc-act, imc-rap}.  It takes into account
the plant vegetation density, water-plant interaction, soil
topography and water-soil interaction.  One note that if the
cover plant is uniform distributed on the soil, there are
not variations in the plant type and plant density, the
model reduces to the classical shallow water equation.  For
numerical applications, we introduce a discrete version of
the model.  First, we use a finite volume method to
discretize the space variable and space derivatives.  The
resulting ordinary differential system of equations sets up
the base of the full discrete model.  Then, we introduce a
fractional time step scheme to discretize the time variable
and time derivatives appearing in the ODE model.  Since
water flow is the main driving force in more complex
phenomena like soil erosion or pollutant spreading, one
needs a water flow modulus to be easily incorporated in a
large computational package.  For this reason, we develop a
simple discrete model with low algebraic calculations and
reduced memory requirements.  For a coupled model water flow
- soil erosion and for the influence of the emergent
vegetation on soil erosion, the reader is referred to
\cite{sander} and \cite{armanini}, respectively.

Acknowledging that there must be a balance between the
accuracy of the model and the computational cost of it, we
underline that we will focus on a first order numerical
method dedicated especially to practical applications on the
field.  This option is a compromise between the precision of
the measured data on the field and the accuracy of the
method.  High order schemes require higher computing
resources and can become prohibitive (due to the enormous
increase of the volume of processed data) when they are
applied at the natural scale of hydrographic basins.

In Section~\ref{sect_ShalowWaterEquations} we introduce the
PDE model; the reader is referred to \cite{imc-act, imc-rap}
for more details concerning this model.
Section~\ref{sect_NumericalApproximation} is dedicated to
the numerical approximation of the solution for this model.
Section~\ref{sect_FVMapproximationof2Dmodel} deals with the
FVM space discretization of the continuum PDE model.  The
method is described in some details to facilitate a less
familiar reader with the FVM method to understand the
discrete version of the continuum model and to allow him to
develop a proper approximation scheme.  In Section
\ref{sect_PropOfSemidiscreteScheme} we investigate some
physical relevant qualitative properties of this ODE system:
monotonicity of the energy, positivity of the water depth
function $h$ and well-balanced properties of the scheme.  In
Section \ref{sect_FractionalSteptimeSchemes} we obtain the
full discrete version of our continuous model.  In fact,
this last version provides the numerical result of the
model.

Section~\ref{sect_ModelValidation} is devoted to the
validation of the numerical method and the continuum PDE
model.  Section~\ref{sect_InternalValidation} deals with the
internal validation.  To analyze the performance of the
numerical scheme, we compare the numerical results with the
exact solutions of PDE in three cases: flow over a bump, the
Tacker solution and the solution of the Riemann problem of
the shallow water equations.  The solutions of Riemann
problem of the shallow water equation with topography and
vegetation require a deep discussion.  If in the case of the
shallow water equation without topography and vegetation
there is no doubt concerning the solution of the problem, in
the case of shallow water equations with topography and
vegetation the definition of solution of the problem require
new concepts, {\it measure solution} and {\it path
  connection}, \cite{maso, sds-riemann-constanta}.  When we
compare a numerical solution with a theoretical one, there
are two relevant aspects for us: the dependence of the
theoretical solution on the definition of the path
connection and the non-uniqueness of the theoretical
solution.

Four more tests are considered in
Section~\ref{sect_ExternalValidation} for external
validation.  The first two tests are performed in order to
compare the numerical results with measured data in
laboratory.  The third one is a qualitative test and
emphasizes the influence of the plant cover on the
asymptotic behavior of the dynamical system defined by the
numerical scheme.  The last test is a numerical simulation
that mimics a real phenomenon, the dumping effect of the
vegetation on the downhill water flow.

Last section is dedicated to final remarks and conclusions.

\section{Mathematical Model}
\label{sect_ShalowWaterEquations}
The model we discuss here is a simplified version a more
general model of water flow on a hillslope introduced in
\cite{imc-rap}.  Let $\Omega\in \mathbb{R}^2$ be a connected
bounded open set.  Assume that the soil surface is
represented by
\begin{equation*}
  x^3=z(x^1,x^2), \quad (x^1,x^2)\in\Omega
\end{equation*}
and the first derivatives of the function $z(\cdot,\cdot)$
are small quantities.  The unknown variables of the model
are the water depth $h(t,x)$ and the two components
$v_a(t,x)$ of the water velocity $\boldsymbol{v}$. The
density of the plant cover is quantified by the porosity
function $\theta:\Omega\rightarrow [0,1]$ which equals to
$1$ for bare soil and $0$ for a complete sealant cover
plant.  The model reads as
\begin{equation}
  \label{swe_vegm_rm.02}
  \begin{split}
    \partial_t(\theta h)+\partial_b(\theta h v^b) & = \mathfrak{M},\\
    \partial_t(\theta hv_a)+\partial_b(\theta
    hv_av^b)+\theta h\partial_aw & =
    \mathfrak{t}^{p}_a+\mathfrak{t}^{s}_a, \quad a=1,2,
  \end{split}
\end{equation}
where
\[ w=g\left[z(x^1,x^2)+h\right] \]
stands for the free water surface level, and $g$ for the
gravitational acceleration.  The contribution of rain and
infiltration to the water mass balance is taken into account
by $\mathfrak{M}$.

The term $\mathfrak{t}^p_a$ quantifies the resistance
opposed by plants to the water flow.  We use an experimental
relation as in \cite{baptist, nepf}
\begin{equation}
  \label{eq_tpa}
  \mathfrak{t}^{p}_a=\alpha_p h \left(1-\theta\right)|\boldsymbol{v}|v_a,
\end{equation}
where $\alpha_p$ is a non-negative constant depending on the
type of vegetation.

The term $\mathfrak{t}^s_a$ quantifies the frictional force
exerted by the interaction water-soil.  In this case, one
can use experimental relations as Manning, Ch\'ezy or
Darcy-Weisbach formula.  The Darcy-Weisbach expression has
the advantage to be non-singular if the water depth becomes
zero.  We analyze the model for
\begin{equation}
  \label{eq_tsa}
  \mathfrak{t}^{s}_a=\theta \alpha_s (h) |\boldsymbol{v}|v_a,
\end{equation}
where $\alpha_s (h)$ is a non-negative function
characteristic to a given soil surface.

Combining the two relations, one can write
\begin{equation}
  \label{eq_tau}
 \mathfrak{t}^{p}_a+\mathfrak{t}^{s}_a  := -{\cal K}(h,\theta)|\boldsymbol{v}|v_a,
\end{equation}
with the function ${\cal K}(h,\theta)$ given by
\begin{equation}
  \label{swe_vegm_rm.03} 
  {\cal K}(h,\theta) = \alpha_p h \left(1-\theta\right) +
  \theta \alpha_s (h).
\end{equation}

Note that there is an energy function ${\cal E}$ given by
\begin{equation}
  \label{swe_vegm_rm.02-0}
  {\cal E} := \frac{1}{2}|\boldsymbol{v}|^2+g\left(x^3+\frac{h}{2}\right),
\end{equation}
satisfying the conservative equation
\begin{equation}
  \label{swe_veg_numerics.08}
  \partial_t (\theta h{\cal E}) +
  \partial_b \left(\theta h v^b \left({\cal E}+g\frac{h}{2}\right)\right) =
  \mathfrak{M}\left(-\frac{1}{2}|\boldsymbol{v}|^2+w\right) -{\cal K}|\boldsymbol{v}|^3.
\end{equation}
In the absence of mass source, a steady state of the system
is given by
\begin{equation}
  \label{swe_veg_numerics.08-01}
  \partial_a(x^3+h)=0, \quad v_a=0, \quad a=1,2,
\end{equation}
whose physical meaning is a lake.

When dealing with constant porosity function $\theta < 1$,
the system (\ref{swe_vegm_rm.02}) is similar to the
classical shallow water equations, while for $\theta = 1$,
we recover the classical shallow water equations.

The model (\ref{swe_vegm_rm.02}) is a hyperbolic system of
equations with source term, see \cite{imc-act}.  Among
different features, we require from our approximation scheme
to have two properties with physical significance:
$h$-positivity and well-balanced (the numerical solution
preserves the lake).  We speculate that the lake is an
attractive steady state of the numerical scheme if one has a
conservative equation for the energy as in the continuous
case.  Schemes with similar properties but in the absence of
vegetation were investigated in \cite{seguin, noelle,
  nordic}.

\section{Numerical Approximation}
\label{sect_NumericalApproximation}

\subsection{Finite Volume Method Approximation of 2D Model}
\label{sect_FVMapproximationof2Dmodel}
Let $\Omega$ be the domain of the space variables $x^1$,
$x^2$ and $\Omega=\cup_i \omega_i, i=\overline{1,N}$ an
admissible polygonal partition, \cite{veque}.  To build a
spatial discrete approximation of the model
(\ref{swe_vegm_rm.02}), one integrates the continuous
equations on each finite volume $\omega _i$ and then defines
an approximation of the integrals.

Let $\omega_i$ be an arbitrary element of the partition.
Relatively to it, the integral form of
(\ref{swe_vegm_rm.02}) reads as \def\msr#1{{\rm m(#1)}}
\begin{equation}
  \label{fvm_2D_eq.01}
  \begin{split}
    \displaystyle\partial_t\int\limits_{\omega_i}\theta h{\rm d}x+
    \int\limits_{\partial\omega_i}\theta h \boldsymbol{v}\cdot\boldsymbol{n}{\rm d}s=
    &\displaystyle\int\limits_{\omega_i}\mathfrak{M}{\rm d}x,\\
      \displaystyle\partial_t\int\limits_{\omega_i}\theta h v_a{\rm d}x+
      \int\limits_{\partial\omega_i}\theta h v_a\boldsymbol{v}\cdot\boldsymbol{n}{\rm d}s+
      \int\limits_{\omega_i}\theta h\partial_a w{\rm d}x=
      &\displaystyle-\int\limits_{\omega_i}{\cal K}|\boldsymbol{v}|v_a{\rm d}x, \quad a=1,2.
  \end{split}
\end{equation}

Now, we build a discrete version of the integral form by
introducing some quadrature formulas.  With $\psi_i$
standing for some approximation of $\psi$ on $\omega_i$, we
introduce the approximations
\begin{equation}
  \label{fvm_2D_eq.01-01}
  \int\limits_{\omega_i}\theta h{\rm d}x\approx\sigma_i\theta_ih_i,\quad
  \int\limits_{\omega_i}\theta h v_a{\rm d}x\approx\sigma_i\theta_ih_iv_{a\,i},\quad
  \int\limits_{\omega_i}{\cal K}|\boldsymbol{v}|v_a{\rm d}x\approx\sigma_i {\cal K}_i|\boldsymbol{v}|_iv_{a\,i},
\end{equation}
where $\sigma_i$ denotes the area of the polygon $\omega_i$.

The integrals of the gradient of the free surface are
approximated by
\begin{equation}
  \label{fvm_2D_eq.01-02}
  \int\limits_{\omega_i}\theta h\partial_a w{\rm
    d}x=
\int\limits_{\partial\omega_i}(w-w_i) \theta hn_a{\rm d}s
 +\mathcal{O}(|\omega_i|h^2)
\approx\int\limits_{\partial\omega_i} (w-w_i) \theta hn_a{\rm d}s,
\end{equation}
where $w_i$ is an approximation of $w$ on $\omega_i$.

Note that if $\omega_i$ is a regular polygon and $w_i$ is
the cell-centered value of $w$, then the approximation is of
second order accuracy for smooth fields and it preserves the
null value in the case of constant fields $w$.

Let $\partial\omega(i,j)$ be the common interface between
the cells $\omega_i$ and $\omega_j$.  We introduce the
notation
\begin{equation}
  \label{fvm_2D_eq.psi}
  \widetilde{\psi}|_{\partial \omega(i,j)}:=\int\limits_{\partial \omega(i,j)}\psi{\rm d}s.
\end{equation}

Using the approximations (\ref{fvm_2D_eq.01-01}) and
(\ref{fvm_2D_eq.01-02}) and keeping the boundary integrals,
one can write
\begin{equation}
  \label{fvm_2D_eq.02}
  \begin{split}
    \sigma_i\partial_t(\theta_i h_i) +
    \sum\limits_{j\in{\cal N}(i)}\widetilde{\theta h v_n}|_{\partial \omega(i,j)}
    &=\sigma_i\mathfrak{M}_i,\\
    \sigma_i\partial_t (\theta_i h_i v_{a\,i}) +
    \sum\limits_{j\in{\cal N}(i)}\widetilde{\theta h v_a v_n}|_{\partial \omega(i,j)}+ 
    \sum\limits_{j\in{\cal N}(i)}\widetilde{(w-w_i)\theta h}n_a|_{\partial \omega(i,j)}
    &=-\sigma_i{\cal K}_i|\boldsymbol{v}|_iv_{a\,i},
  \end{split}
\end{equation}
where ${\cal N}(i)$ denotes the set of indexes of all the
neighbors of $\omega_i$.

The next step is to define the approximations of the
boundary integrals in (\ref{fvm_2D_eq.02}). We approximate
an integral $\widetilde{\psi}|_{\partial \omega(i,j)}$ of
the form (\ref{fvm_2D_eq.psi}) by considering the integrand
$\psi$ to be a constant function
$\psi_{(i,j)}(\psi_i,\psi_j)$, where $\psi_i$ and $\psi_j$
are some fixed values of $\psi$ on the adjacent cells
$\omega_i$ and $\omega_j$, respectively.  Thus,
\begin{equation}
  \label{fvm_2D_eq.03}
  \begin{array}{l}
    \widetilde{\theta h v_n}|_{\partial \omega(i,j)}\approx l_{(i,j)}\theta h_{(i,j)} (v_n)_{(i,j)},\\
    \widetilde{\theta h v_a v_n}|_{\partial \omega(i,j)}\approx l_{(i,j)}\theta h_{(i,j)} (v_a)_{(i,j)}(v_n)_{(i,j)},\\
    \widetilde{(w-w_i)\theta h}n_a|_{\partial \omega(i,j)} \approx l_{(i,j)}(w_{(i,j)}-w_i) \theta h^s_{(i,j)}(n_a)_{(i,j)},
   \end{array}
\end{equation} 
where $\boldsymbol{n}_{(i,j)}$ denotes the unitary normal to
the common side of $\omega_i$ and $\omega_j$ pointing
towards $\omega_j$, and $l_{(i,j)}$ is the length of this
common side.

The issue is to define the interface value functions
$\psi_{(i,j)}(\psi_i,\psi_j)$ such that the resulting scheme
has certain desired properties.  Among them, we impose the
scheme to be well-balanced and to preserve the positivity of
the water depth function $h$.

\medskip\noindent
{\bf Well-balanced and h-positive scheme.}  For any
internal interface $(i,j)$, we define the following
quantities:
\begin{equation}
  \label{fvm_2D_eq.04}
  \begin{array}{l}
    (v_a)_{(i,j)}=\displaystyle\frac{v_{a\,i}+v_{a\,j}}{2}, \quad a=1,2,\\
    (v_n)_{(i,j)}=\boldsymbol{v}_{(i,j)}\cdot\boldsymbol{n}_{(i,j)},\\
    w_{(i,j)}=\displaystyle\frac{w_{i}+w_{j}}{2},
  \end{array}
\end{equation}
and 
\begin{equation}
  \label{fvm_2D_eq.06}
  \theta h_{(i,j)}=
  \left\{
    \begin{array}{ll}
      \theta_i h_i, & {\rm if}\; (v_n)_{(i,j)}>0,\\
      \theta_j h_j, & {\rm if}\; (v_n)_{(i,j)}<0. 
    \end{array}
  \right.
\end{equation}

For the term $(\theta h)^s_{(i,j)}$, we analyze the
alternatives
\begin{equation}
  \label{fvm_2D_eq.05}
  (\theta h)^s_{(i,j)} =
  \left\{
    \begin{array}{ll}
      \theta h_{(i,j)}, & {\rm if}\; (v_n)_{(i,j)}\neq 0\\
      \theta_i h_i, & {\rm if}\; (v_n)_{(i,j)}=0 \;{\rm and}\; w_i>w_j\\
      \theta_j h_j, & {\rm if}\; (v_n)_{(i,j)}=0 \;{\rm and}\; w_i\leq w_j
    \end{array},
  \right.
\end{equation}
\begin{equation}
  \label{fvm_2D_eq.05_varianta2}
  (\theta h)^s_{(i,j)} =
  \left\{
    \begin{array}{ll}
      \theta_i h_i, & {\rm if}\; w_i>w_j\\
      \theta_j h_j, & {\rm if}\; w_i\leq w_j
    \end{array}
  \right.
\end{equation}
or  
\begin{equation}
  \label{fvm_2D_eq.05_varianta3}
  (\theta h)^s_{(i,j)} = \frac{\theta_i h_i+\theta_j h_j}{2}.
\end{equation}

The semidiscrete scheme takes now the form of the following
system of ODEs
\begin{equation}
  \label{fvm_2D_eq.07}
  \begin{split}
    \sigma_i\displaystyle\frac{\rm d}{{\rm d}t}
    (\theta_i h_i)+\sum\limits_{j\in{\cal N}(i)}l_{(i,j)}\theta h_{(i,j)} (v_n)_{(i,j)}&=\sigma_i\mathfrak{M}_i,\\
    \sigma_i\displaystyle\frac{\rm d}{{\rm d}t}
    (\theta_i h_i v_{a\,i})
    +\sum\limits_{j\in{\cal N}(i)}l_{(i,j)}\theta h_{(i,j)} (v_a)_{(i,j)} (v_n)_{(i,j)}+&\\ 
    +\displaystyle\frac{1}{2}\sum\limits_{j\in{\cal N}(i)}l_{(i,j)}(w_j-w_i)(\theta h)^s_{(i,j)} n_a|_{(i,j)}&=-\sigma_i{\cal K}_i|\boldsymbol{v}|_iv_{a\,i}.
  \end{split}
\end{equation}

\medskip\noindent {\bf Boundary conditions. Free discharge.}
We need to define the values of $h$ and $\boldsymbol{v}$ on
the external sides of $\Omega$.  For each side in
$\Gamma=\partial\Omega$, we introduce a new cell (``ghost''
element) adjacent to the polygon $\omega_i$ corresponding to
that side.  For each ``ghost'' element, one must somehow
define its altitude and then we set zero values to its water
depth.  We can now define $h$ and $\boldsymbol{v}$ on the
external sides of $\Omega$ by
\begin{equation}
  \label{fvm_2D_eq.070}
  \begin{array}{l}
    \boldsymbol{v}_{\partial \omega_i\cap \Gamma}=\boldsymbol{v}_i,\\
    h_{\partial \omega_i\cap \Gamma}=
    \left\{
    \begin{array}{ll}
      h_i, & {\rm if}\; \boldsymbol{v}_i\cdot \boldsymbol{n}|_{\partial \omega_i\cap \Gamma }>0,\\
      0, & {\rm if}\; \boldsymbol{v}_i\cdot \boldsymbol{n}|_{\partial \omega_i\cap \Gamma }<0.
    \end{array}
    \right.
  \end{array}
\end{equation}
Now, the solution is sought inside the positive cone
$h_i>0, \; i=\overline{1,N}$.

\subsection{Properties of the Semidiscrete Scheme}
\label{sect_PropOfSemidiscreteScheme}
The ODE model (\ref{fvm_2D_eq.07}) is the basis of the
numerical scheme for obtaining a numerical solution of the
PDE model.  Therefore, it is worthwhile to analyze the
properties of (\ref{fvm_2D_eq.07}).  Unfortunately, not all
the properties of the solution of the ODE model can be
preserved by its numerical approximation.  Numerical schemes
preserving properties of some particular solutions of the
continuum model were and are intensively investigated in the
literature \cite{seguin, nordic, bouchut-book,
  well-balanced}.  We analyze such properties for the
semi-discrete scheme in the present section and dedicate the
next section to the properties of the complete discretized
scheme.  Depending on the scheme used to approximate the
terms $(\theta h)^s_{(i,j)}$, the ODE model
(\ref{fvm_2D_eq.07}) has different properties.  If one uses
the formula (\ref{fvm_2D_eq.05}), then the ODE model can
have discontinuities in r.h.s. and therefore it is possible
that the solution in the classical sense of this system
might not exist for some initial data.  However, the
solution in Filipov sense \cite{filipov} exists for any
initial data.  There are initial data for which the solution
in the classical sense exists only locally in time.  As a
reward, there is an energy balance equation and the energy
is a monotone function with respect to time.  On the other
hand, if the alternative formula
(\ref{fvm_2D_eq.05_varianta2}) or
(\ref{fvm_2D_eq.05_varianta3}) is used, then the r.h.s. of
the ODE model is defined by continuum functions, but we can
not anymore find an energetic function.

\subsubsection{Energy Balance}
Definitions (\ref{fvm_2D_eq.04}) and (\ref{fvm_2D_eq.05})
yields a dissipative equation in conservative form for the
cell energy ${\cal E}_i$,
\begin{equation}
  \label{fvm_2D_eq.060}
  {\cal E}_i(h_i,\boldsymbol{v}_i)=\theta_i\left(\frac{1}{2}{|\boldsymbol{v}|^2_i}{h_i}+\frac{1}{2}gh^2_i+gx^3_ih_i\right). 
\end{equation}
The time derivative of ${\cal E}_i$ can be written as
\begin{equation}
  \label{fvm_2D_eq.0700}
  \sigma_i \displaystyle\frac{{\rm d }{{\cal E}_i} }{ {\rm d} t} = \sigma_i
  \left(
    \left(w_i-\frac{1}{2}|\boldsymbol{v}|^2_i\right) 
\displaystyle\frac{{\rm d }{(\theta_i h_i)}}{ {\rm d} t} 
    +\left< \boldsymbol{v}_i, \displaystyle\frac{{\rm d }{(\theta_ih_i\boldsymbol{v}_i)}}{ {\rm d} t} \right>
  \right),
\end{equation}
where $\left<\cdot,\cdot\right>$ denotes the euclidean
scalar product in $\mathbb{R}^2$.

\begin{proposition}[Cell energy equation]
  \label{cell_energy}
  In the absence of mass source, one has
  \begin{equation}
    \label{fvm_2D_eq.071}
    \sigma_i\displaystyle\frac{\rm d}{{\rm d}t}{\cal E}_i+\sum\limits_{j\in{\cal N}(i)}l_{(i,j)}\left<{\cal H}_{(i,j)},\boldsymbol{n}_{(i,j)}\right>=-\sigma_i{\cal K}_i|\boldsymbol{v}|^3_i,
  \end{equation}
  where
  \begin{equation*}
    {\cal H}_{(i,j)}=\frac{1}{2}\theta h_{(i,j)}
    \left(
      w_i\boldsymbol{v}_i+w_j\boldsymbol{v}_j+\left<\boldsymbol{v}_i,\boldsymbol{v}_j\right>\boldsymbol{v}_{(i,j)}
    \right).
  \end{equation*}
\end{proposition}

\begin{remark}
  If $(\theta h,v,w)_j=(\theta h,v,w)_i$ for any
  $j\in{\cal N}(i)$, then
  \begin{equation*}
    {\cal H}=\theta h \boldsymbol{v}\left(\frac{1}{2} |\boldsymbol{v}|^2+w\right)
  \end{equation*}
  is the continuous energy flux in {\rm
    (\ref{swe_veg_numerics.08})}.
\end{remark}

\begin{proof}
Using the equality
(\ref{fvm_2D_eq.0700}), we can write
\begin{equation*}
  \begin{array}{rcl}
    \sigma_i\displaystyle\frac{\rm d}{{\rm d}t}{\cal E}_i&=&-(w_i-\displaystyle\frac{1}{2}|\boldsymbol{v}|_i^2)\sum\limits_{j\in{\cal N}(i)}l_{(i,j)}\theta h_{(i,j)} (v_n)_{(i,j)}-\\
    &&-\left<
    \boldsymbol{v}_i,\sum\limits_{j\in{\cal N}(i)}l_{(i,j)}\theta h_{(i,j)} \boldsymbol{v}_{(i,j)} (v_n)_{(i,j)}
      \right>-\\
    &&-\displaystyle\frac{1}{2}\left<\boldsymbol{v}_i,\sum\limits_{j\in{\cal N}(i)}l_{(i,j)}(w_j-w_i)(\theta h)^s_{(i,j)} \boldsymbol{n}_{(i,j)}\right>-\\
    &&-\sigma_i{\cal K}_i|\boldsymbol{v}|_i^3.
  \end{array}
\end{equation*}
Now, one has the identities
\begin{equation*}
  \begin{array}{rcl}
    w_i\displaystyle\sum\limits_{j\in{\cal N}(i)}l_{(i,j)}\theta h_{(i,j)} (v_n)_{(i,j)}
    &=&\displaystyle\sum\limits_{j\in{\cal N}(i)}l_{(i,j)}\theta h_{(i,j)} (v_n)_{(i,j)}\displaystyle\frac{w_i+w_j}{2}+\\
    &&+\displaystyle\sum\limits_{j\in{\cal N}(i)}l_{(i,j)}\theta h_{(i,j)} (v_n)_{(i,j)}\displaystyle\frac{w_i-w_j}{2}
  \end{array}
\end{equation*}
and
\begin{equation*}
  \begin{array}{l}
    \left<\boldsymbol{v}_i,\displaystyle\sum\limits_{j\in{\cal N}(i)}l_{(i,j)}(w_j-w_i)(\theta h)^s_{(i,j)} \boldsymbol{n}_{(i,j)}\right>=\\
    =\displaystyle\sum\limits_{j\in{\cal N}(i)}l_{(i,j)}(w_j-w_i)(\theta h)^s_{(i,j)}\left<\displaystyle\frac{\boldsymbol{v}_i+\boldsymbol{v}_j}{2}+\displaystyle\frac{\boldsymbol{v}_i-\boldsymbol{v}_j}{2}, \boldsymbol{n}_{(i,j)}\right>.
  \end{array}
\end{equation*}
Therefore,
\begin{equation*}
  \begin{array}{r}
    w_i\displaystyle\sum\limits_{j\in{\cal N}(i)}l_{(i,j)}\theta h_{(i,j)} (v_n)_{(i,j)}+
    \displaystyle\frac{1}{2}\left<\boldsymbol{v}_i,\displaystyle\sum\limits_{j\in{\cal N}(i)}l_{(i,j)}(w_j-w_i)(\theta h)^s_{(i,j)} \boldsymbol{n}_{(i,j)}\right>=\\
    =\displaystyle\sum\limits_{j\in{\cal N}(i)}l_{(i,j)}\theta h_{(i,j)}\left<w_i\boldsymbol{v}_i+w_j\boldsymbol{v}_j,\boldsymbol{n}_{(i,j)}\right>.
  \end{array}
\end{equation*}
Similarly, one obtains the identity
\begin{equation*}
  \begin{array}{r}
    -\displaystyle\frac{1}{2}|\boldsymbol{v}|_i^2\sum\limits_{j\in{\cal N}(i)}l_{(i,j)}\theta h_{(i,j)} (v_n)_{(i,j)}+
    \left<\boldsymbol{v}_i,\sum\limits_{j\in{\cal N}(i)}l_{(i,j)}\theta h_{(i,j)} \boldsymbol{v}_{(i,j)} (v_n)_{(i,j)}\right>=\\
    =\displaystyle\frac{1}{2}\sum\limits_{j\in{\cal N}(i)}l_{(i,j)}\theta h_{(i,j)}\left<\boldsymbol{v}_i,\boldsymbol{v}_j\right>\left<\displaystyle\frac{\boldsymbol{v}_i+\boldsymbol{v}_j}{2}, \boldsymbol{n}_{(i,j)}\right>.
  \end{array}
\end{equation*}
\end{proof}
Taking out the mass exchange through the boundary, the
definitions of the interface values ensure the monotonicity
of the energy with respect to time.

\subsubsection{h-positivity and Critical Points}
\begin{proposition}[h-positivity]
  If $\mathfrak{M} \geq 0 $ when $h=0$, then the ODE system
  {\rm (\ref{fvm_2D_eq.07})} with {\rm
    (\ref{fvm_2D_eq.04})}, {\rm (\ref{fvm_2D_eq.05})}, {\rm
    (\ref{fvm_2D_eq.06})} and {\rm (\ref{fvm_2D_eq.070})}
  preserves the positivity of the water depth function $h$.
\end{proposition}

\begin{proof}
One can rewrite the mass balance equations as
\begin{equation*}
  \sigma_i\displaystyle\frac{\rm d}{{\rm d}t}(\theta_i h_i)=
  -\theta_i h_i \sum\limits_{j\in{\cal N}(i)}l_{(i,j)} (v_n)^{+}_{(i,j)}+
  \sum\limits_{j\in{\cal N}(i)}l_{(i,j)}
  \theta_j h_j (v_n)^{-}_{(i,j)} + \sigma_i \mathfrak{M}_i.
\end{equation*}
Observe that if $h_i=0$ for some $i$, then
$\sigma_i\displaystyle\frac{\rm d}{{\rm d}t}
(\theta_i h_i)\geq 0$.
\end{proof}

There are two kinds of stationary points for the ODE model:
the lake and uniform flow on an infinitely extended plan
with constant vegetation density.

\begin{proposition}[Stationary point. Uniform flow.]
  \label{river}
  Consider $\{\omega_i\}_{i=\overline{1,N}}$ to be a regular
  partition of $\Omega$ with $\omega_i$ regular polygons.
  Let $z-z_0=\xi_b x^b$ be a representation of the soil
  plane surface.  Assume that the discretization of the soil
  surface is given by
  \begin{equation}
    \label{fvm_2D_eq.08}
    z_i-z_0=\xi_b \overline{x}^b_i,
  \end{equation}
  where $\overline{x}^b_i$ is the mass center of the
  $\omega_i$ and $\theta_i=\theta$.  Then, given a value
  $h$, there is $\boldsymbol{v}$ so that the state
  $(h_i,\boldsymbol{v}_i)=(h,\boldsymbol{v})$,
  $i=\overline{1,N}$ is a stationary point of the ODE {\rm
    (\ref{fvm_2D_eq.07})}.
\end{proposition}

\begin{proof}
  For any constant state $h_i=h$ and $(v_a)_i=v_a$, the ODE
  (\ref{fvm_2D_eq.07}) reduces to
  \begin{equation*}
    \displaystyle\frac{1}{2}\theta h g\sum\limits_{j\in{\cal
        N}(i)}l_{(i,j)}(z_j-z_i) n_a|_{(i,j)} =-\sigma{\cal
      K}|\boldsymbol{v}|v_{a}.
  \end{equation*}
  Introducing the representation (\ref{fvm_2D_eq.08}), one
  writes
  \begin{equation*}
    \displaystyle\frac{1}{2}\theta h g\sum\limits_{j\in{\cal N}(i)}l_{(i,j)}\xi_b(\overline{x}^b_j-\overline{x}^b_i) n_a|_{(i,j)}
    =-\sigma{\cal K}|\boldsymbol{v}|v_{a}.
  \end{equation*}
  Note that for a regular partition, one has the identity
  \begin{equation*}
    \overline{x}^b_j-\overline{x}^b_i=2(y_{(i,j)}-\overline{x}^b_i),
  \end{equation*}
  where $y_{(i,j)}$ is the midpoint of the common side
  $\overline{\omega}_i\cap\overline{\omega}_j$.  Taking into
  account
  \begin{equation*}
    \begin{split}
      \displaystyle\frac{1}{2}\theta h g
      \sum\limits_{j\in{\cal N}(i)}l_{(i,j)}(z_j-z_i)
      n_a|_{(i,j)} &=\displaystyle\theta h g
                     \sum\limits_{j\in{\cal N}(i)}l_{(i,j)}\xi_b y^b_{(i,j)} n_a|_{(i,j)}\\
                   &=\displaystyle\theta h g
                     \int\limits_{\partial \omega_i}\xi_b x^b(s) n_a(s){\rm d}s\\
                   &=\displaystyle\theta h g
                     \int\limits_{\omega_i}\xi_b \partial_a x^b{\rm d}x\\
                   &=\sigma \theta h g \xi_a,
    \end{split}
  \end{equation*}
  we obtain that the velocity is a constant field
  \begin{equation}
    \label{fvm_2D_eq.09}
    v_a=\xi_a\left(\frac{\theta h g}{{\cal K} |\xi|}\right)^{1/2}.
  \end{equation}
\end{proof}

A lake is a stationary point characterized by a constant
value of the free surface and a null velocity field over
connected regions.  A lake for which $h_i>0$ for any
$i\in\{1,2,\ldots, N\}$ will be named {\it regular
  stationary point} and a lake that occupies only a part of
a domain flow will be named {\it singular stationary point}.

\begin{proposition}[Stationary point. Lake.]
  \label{lake}
  Denote by $\bar{w}$ some real number value.  In the
  absence of mass source, the following properties hold:

  {\rm (a)} Regular stationary point: the state
  \begin{equation*}
    w_i = \bar{w} \;\; \& \;\; \boldsymbol{v}_i=0, \; \forall i=\overline{1,N}
  \end{equation*}
  is a stationary point of ODE {\rm (\ref{fvm_2D_eq.07})}.

  {\rm (b)} Singular stationary point: if
  $(\theta h)^s_{(i,j)}$ is calculated through {\rm
    (\ref{fvm_2D_eq.05})} or {\rm
    (\ref{fvm_2D_eq.05_varianta2})}, then the state
  \begin{equation*}
    \boldsymbol{v}_i=0, \; \forall i=\overline{1,N} \;\; \& \;\;
    w_i = \bar{w}, \; \forall i\in {\cal I} \;\; \& \;\; h_i=0, \;
    z_i>w, \; \forall i\in \complement{\cal I},
  \end{equation*}
  for some ${\cal I}\subset \{1,2,\ldots,N\}$ is a
  stationary point. ($\complement{\cal I}$ is the complement
  of ${\cal I}$.)
\end{proposition}

\begin{proof}
  For the sake of simplicity, in the case of the singular
  stationary point, we consider that
  $\displaystyle\Omega_{\cal I}=\cup_{i\in{\cal I}}\omega_i$
  is a connected domain.  Since the velocity field is zero,
  it only remains to verify that
  \begin{equation*}
    \sum\limits_{j\in{\cal N}(i)}l_{(i,j)}(w_j-w_i)(\theta
    h)^s_{(i,j)} n_a|_{(i,j)}=0,
  \end{equation*}
  for any cell $\omega_i$.  If $i \in \complement{\cal I}$,
  then the above sum equals zero since $h^s_{(i,j)}=0$, for
  all $j\in{\cal N}(i)$.  If $i \in {\cal I}$, then the sum
  is again zero because either $h^s_{(i,j)}=0$, for
  $j\in\complement{\cal I}$ or $w_j=w_i$, for
  $j\in{\cal I}$.
\end{proof}

\subsection{Time Integration Scheme}
\label{sect_FractionalSteptimeSchemes}
In what follows, we introduce a (first order) fractional
time integration scheme in order to integrate the ODE
(\ref{fvm_2D_eq.07}).

We introduce some notations
\begin{equation}
  \label{fvm_2D_eq_frac.01}
  \begin{split}
    {\cal J}_{a\,i}(h,\boldsymbol{v}):=&-\displaystyle\sum\limits_{j\in{\cal N}(i)}l_{(i,j)}\theta h_{(i,j)} (v_a)_{(i,j)} (v_n)_{(i,j)},\\
    {\cal S}_{a\,i}(h,w):=&-\displaystyle\frac{1}{2}\sum\limits_{j\in{\cal N}(i)}l_{(i,j)}(w_j-w_i)(\theta h)^{s}_{(i,j)} n_a|_{(i,j)},\\
    {\cal L}_i((h,\boldsymbol{v})):=&-\displaystyle\sum\limits_{j\in{\cal N}(i)}l_{(i,j)}\theta h_{(i,j)} (v_n)_{(i,j)}.
  \end{split}
\end{equation}

Now, (\ref{fvm_2D_eq.07}) becomes
\begin{equation}
  \label{fvm_2D_eq_frac.02}
  \begin{split}
    \sigma_i\displaystyle\frac{\rm d}{{\rm d}t}
(\theta_i h_i)&={\cal L}_i(h,\boldsymbol{v})+\sigma_i\mathfrak{M}(t,h),\\
    \sigma_i\displaystyle\frac{\rm d}{{\rm d}t}
(\theta_i h_i v_{a\,i})&={\cal J}_{a\,i}(h,\boldsymbol{v})+{\cal S}_{a\,i}(h,w)-{\cal K}(h)|\boldsymbol{v}_i| v_{a\,i}.
  \end{split}
\end{equation}

\noindent {\bf Mass source.} We assume that the mass source
$\mathfrak{M}$ is of the form
\begin{equation}
  \label{fvm_2D_eq.14}
  \mathfrak{M}(x,t,h)=r(t)-\theta(x)\iota(t,h),
\end{equation}
where $r(t)$ quantifies the rate of the rain and
$\iota(t,h)$ quantifies the infiltration rate.  The
infiltration rate is a continuous function and satisfies the
following condition
\begin{equation}
  \label{fvm_2D_eq.15}
  \iota(t,h)<\iota_{m},\quad {\rm if}\; h\geq 0.
\end{equation}
The basic idea of a fractional time method is to split the
initial ODE into two sub-models, integrate them separately,
and then combine the two solutions \cite{veque-phd, strang}.

We split the ODE (\ref{fvm_2D_eq.07}) into
\begin{equation}
  \label{fvm_2D_eq_frac.03}
  \begin{split}
    \sigma_i\displaystyle\frac{\rm d}{{\rm d}t}
(\theta_i h_i)&={\cal L}_i(h,\boldsymbol{v}),\\
    \sigma_i\displaystyle\frac{\rm d}{{\rm d}t}
(\theta_i h_i v_{a\,i})&={\cal J}_{a\,i}(h,\boldsymbol{v}) +{\cal S}_{a\,i}(h,w),
  \end{split}
\end{equation}
and
\begin{equation}
  \label{fvm_2D_eq_frac.04}
  \begin{split}
    \sigma_i\displaystyle\frac{\rm d}{{\rm d}t}
(\theta_i h_i)&=\sigma_i\mathfrak{M}_i(t,h),\\
    \sigma_i\displaystyle\frac{\rm d}{{\rm d}t}
(\theta_i h_i v_{a\,i})&=-{\cal K}(h)|\boldsymbol{v}_i| v_{a\,i}.
  \end{split}
\end{equation}
A first order fractional step time accuracy reads as
\begin{equation}
  \label{fvm_2D_eq_frac.05}
  \begin{split}
    \sigma(\theta h)^{*}&=\sigma(\theta h)^n+\triangle t_n{\cal L}((h,\boldsymbol{v})^{n}),\\
    \sigma(\theta hv_a)^{*}&=\sigma(\theta hv_a)^{n}+\triangle t_n \left({\cal J}_a((h,\boldsymbol{v})^{n})+{\cal S}_a((h,w)^{n})\right),\\
  \end{split}
\end{equation}
\begin{equation}
  \label{fvm_2D_eq_frac.06}
  \begin{split}
    \sigma(\theta h)^{n+1}&=\sigma(\theta h)^{*}+\sigma\triangle t_n\mathfrak{M}(t^{n+1},h^{n+1}),\\
    \sigma(\theta hv_a)^{n+1}&=\sigma(\theta hv_a)^{*}-\triangle t_n{\cal K}(h)|\boldsymbol{v}^{n+1}| v^{n+1}_{a}.\\
  \end{split}
\end{equation}
Steps (\ref{fvm_2D_eq_frac.05}) and
(\ref{fvm_2D_eq_frac.06}) lead to
\begin{equation}
  \label{fvm_2D_eq_frac.07}
  \begin{array}{rcl}
    \sigma(\theta h)^{n+1}&=&\sigma(\theta h)^n+\triangle t_n{\cal L}((h,\boldsymbol{v})^{n})+\sigma\triangle t_n\mathfrak{M}(t^{n+1},h^{n+1}),\\
    \sigma(\theta hv_a)^{n+1}&=&\sigma(\theta hv_a)^{n}+\triangle t_n\left({\cal J}_a((h,\boldsymbol{v})^{n})+{\cal S}_a((h,w)^{n})\right)-\\
    &&-\triangle t_n\sigma{\cal K}(h)|\boldsymbol{v}^{n+1}| v^{n+1}_{a}.
  \end{array}
\end{equation} 
To advance a time step, one needs to solve a scalar
nonlinear equation for $h$ and a 2D nonlinear system of
equations for velocity $\boldsymbol{v}$.

In what follows, we investigate some important physical
properties of the numerical solution given by
(\ref{fvm_2D_eq_frac.07}): $h$-positivity, well-balanced
property and monotonicity of the energy.

\subsubsection{h-positivity. Stationary Points}
\begin{proposition}[$h$-positivity]
  There is an upper bound $\tau_n$ for the time step
  $\triangle t_n$ such that if $\triangle t_n<\tau_n$ and
  $h^n>0$, then $h^{n+1}\geq 0$.
\end{proposition}

\noindent
\begin{proof} For any cell $i$ one has
\begin{equation*}
  \begin{split}
    \sigma_i\theta_ih^{n+1}_i+\triangle t_n\iota(t^{n+1},h^{n+1}_i)=&\sigma_i\theta_ih^{n}_i\left(1-\displaystyle\frac{\triangle t_n}{\sigma_i}\sum\limits_{j\in{\cal N}(i)}l_{(i,j)} (v_n)^{n,+}_{(i,j)}
      \right)+\\
    &\displaystyle +\triangle t_n \sum\limits_{j\in{\cal N}(i)}l_{(i,j)}(\theta h^n)_{j} (v_n)^{n,-}_{(i,j)}+\triangle t_nr(t^{n+1}).
  \end{split}
\end{equation*}
A choice for the upper bound $\tau_n$ is given by
\begin{equation}
  \label{fvm_2D_eq_frac.07-1}
  \tau_n=\displaystyle\frac{1}{v^n_{\rm max}}\min_i\left\{\displaystyle\frac{\sigma_i}{\sum\limits_{j\in{\cal N}(i)}l_{(i,j)}}\right\}.
\end{equation}
\end{proof}

\begin{proposition}[Well-balanced]
  The lake and the uniform flow are stationary points of
  scheme {\rm (\ref{fvm_2D_eq_frac.07})}.
\end{proposition}

\noindent
\begin{proof}
  One can prove this result similarly as in propositions
  \ref{river} and \ref{lake}.
\end{proof}

Unfortunately, the semi-implicit scheme
(\ref{fvm_2D_eq_frac.07}) does not preserve the monotonicity
of the energy.

\subsubsection{Discrete Energy}   
The variation of the energy between two consecutive time
steps can be written as
\begin{equation}
  \label{fvm_2D_eq_frac.08}
  \begin{split}
    {\cal E}^{n+1}-{\cal E}^n=&\displaystyle\sum\limits_i\theta_i\sigma_i(h_i^{n+1}-h_i^n)(w_i^n-\displaystyle\frac{|\boldsymbol{v}^n_i|^2}{2})+\\
    &+\displaystyle\sum\limits_i\theta_i\sigma_i\left<(h\boldsymbol{v})^{n+1}_i-(h\boldsymbol{v})^{n}_i,\boldsymbol{v}^n_i\right>+\\
    &+g\displaystyle\sum\limits_i\theta_i\sigma_i\displaystyle\frac{(h_i^{n+1}-h_i^n)^2}{2}+\sum\limits_i\theta_i\sigma_i \displaystyle\frac{h_i^{n+1}}{2}\left|\boldsymbol{v}^{n+1}_i-\boldsymbol{v}^{n}_i\right|^2.
  \end{split}
\end{equation}
If the sequence $(h,\boldsymbol{v})^n$ is given by the
scheme (\ref{fvm_2D_eq_frac.07}), we obtain
\begin{equation}
  \label{fvm_2D_eq_frac.09}
  \begin{split}
    {\cal E}^{n+1}-{\cal E}^n=&-\displaystyle\triangle t_n\sum\limits_i\sigma_i{\cal K}(h^{n+1})|\boldsymbol{v}^{n+1}_i|^2\left<\boldsymbol{v}_i^{n+1},\boldsymbol{v}^{n}_i\right>+\\
    &+\displaystyle g\sum\limits_i\theta_i\sigma_i\displaystyle\frac{(h_i^{n+1}-h_i^n)^2}{2}+\sum\limits_i\theta_i\sigma_i\displaystyle\frac{h_i^{n+1}}{2}\left|\boldsymbol{v}^{n+1}_i-\boldsymbol{v}^{n}_i\right|^2+\\
    &+TS+TB,
  \end{split}
\end{equation}                          
where $TB$ and $TS$ stand for the contribution of boundary
and mass source to the energy production.

\subsubsection{Stability}
The stability of any numerical scheme ensures that errors in
data at a time step are not further amplified along the next
steps. To acquire the stability of our scheme, we have
investigated several time-bounds $\tau_n$ and we introduced
an artificial viscosity in the scheme to force the monotony
of the discrete energy.
  
The artificial viscosity is a known method for improving the
performance of a numerical scheme, see for example
\cite{veque, kurganov}.  So, we add a ``viscous'' term to
${\cal J}$:
\begin{equation}
  \label{eq_artif_visc}
  {\cal J}^{\boldsymbol{v}}_{a\,i}={\cal J}_{a\,i}(h,\boldsymbol{v})+
  \displaystyle\sum\limits_{j\in{\cal N}(i)}l_{(i,j)} \nu_{a(i,j)}(h,\boldsymbol{v}).
\end{equation}

In our numerical experiments for the case of continuum soil
and porosity functions, we found out that the following
version of $\nu_{a(i,j)}(h,\boldsymbol{v})$ works very well:
\begin{equation}
  \label{eq_visc1}
  \nu_{a(i,j)}(h,\boldsymbol{v}) =c_{(i,j)}\mu_{(i,j)}
  \left( (v_a)_j - (v_a)_i \right),
\end{equation}
where 
\begin{equation}
  c_i=|\boldsymbol{v}|_i+\sqrt{gh_i}, \quad  c_{(i,j)}=\max\{c_i,c_j\}
\end{equation}
and 
\begin{equation}
  \mu_{(i,j)}=\displaystyle
 \frac{2 \theta_i h_i \theta_j h_j}{\theta_i h_i +
    \theta_j h_j}.
\end{equation}
The case with discontinuous data will be discussed later.

The variation of energy is now given by
\begin{equation}
  \label{fvm_2D_eq_frac.11}
  {\cal E}^{n+1}_{\boldsymbol{v}}-{\cal E}_{\boldsymbol{v}}^n=
  {\cal E}^{n+1}-{\cal E}^n-
  \triangle t_n\displaystyle\sum\limits_{s(i,j)}l_{(i,j)}
  \mu_{(i,j)}c_{(i,j)}
  \left|\boldsymbol{v}_i-\boldsymbol{v}_j\right|^2.
\end{equation}

The hyperbolic character of the shallow water equations and
the necessity to preserve the positivity of the water depth
impose an upper bound for the time step.  Both requirements
are satisfied if this upper bound is chosen to be
\begin{equation}
  \label{fvm_2D_eq_frac.12}
  \tau_n=CFL\displaystyle\frac{\phi_{\rm min}}{c^n_{\rm max}},
\end{equation}
where
\begin{equation}
  \label{fvm_2D_eq_frac.13}
  c_{\rm max}=\max\limits_i \{c_i\},\quad
  \phi_{\rm min}=\min\limits_i
  \left\{
    \displaystyle\frac{\sigma_i}{\sum\limits_{j\in{\cal N}(i)}l_{(i,j)}}
  \right\}
\end{equation}
and $CFL$ is a number between $0$ and $1$.
\begin{remark}
  An upper bound of type {\rm (\ref{fvm_2D_eq_frac.12})} for
  the time-step is well known in the theory of hyperbolic
  systems as the Courant–Friedrichs–Lewy condition {\rm
    \cite{veque, bouchut-book}}.
\end{remark}

\section{Model Validation}
\label{sect_ModelValidation}
A rough classification of validation methods splits them
into two classes: internal and external.  For the internal
validation, one analyses the numerical results into a
theoretical frame: comparison to analytic results,  
sensibility to the variation of the parameters, robustness,
stability with respect to the errors in the input data etc.
These methods validate the numerical results with respect to
the mathematical model and not with the physical processes;
this type of validation is absolutely necessary to ensure
the mathematical consistency of the method.

The external validation methods assume a comparison of the
numerical data with measured real data.  The main advantage
of these methods is that a good consistency of data
validates both the numerical data and the mathematical
model.  In the absence of measured data, one can do a
qualitative analysis: the evolution given by the numerical
model is similar to the observed one, without pretending
quantitative estimations.

\subsection{Internal Validation}
\label{sect_InternalValidation}
We analyze the performance of the 1D version of our
numerical scheme by considering the following test problems:
\begin{enumerate}
  \item flow over a bump, 
  \item Thacker problem, 
  \item Riemann problems.
\end{enumerate}

\subsubsection{Flow Over a Bump}
\label{sect_FlowOveraBump}
We consider two classical tests related to the steady
solution and the propagation of a wave over a bump.

\medskip
\underline{Steady solution}
\medskip

The steady solution was analyzed in \cite{fayssal, delestre,
  zhou}.  In the absence of all frictional terms and mass
sources, a steady solution of the shallow water equations is
given by the solution of the algebraic equations
\begin{equation}
  \label{validatation.eq-01}
  \left\{
    \begin{array}{l}
      hu=q_0\\
      \displaystyle\frac{u^2}{2}+gh(z+h)=\xi_0
    \end{array}
  \right. .
\end{equation}
We consider the case of a soil surface that is flat with a
bump
\begin{equation}
  \label{validatation.eq-02}
  z(x) =
  \left\{
    \begin{array}{ll}
      0.2-0.05(x-12.5)^2, & {\rm if}\; |x-12.5|<2\\
      0,&  {\rm otherwise}
    \end{array}
  \right. ,
\end{equation}
defined over the domain $\Omega=[0,L]$ of the flow with
$L=25$m.  The only boundary condition we use is
\begin{equation}
  \label{validatation.eq-03}
  \left. hu \right|_{x=0} = q_0.
\end{equation}
The initial state is given by
\begin{equation}
  \label{validatation.eq-04}
  \left\{
    \begin{array}{l}
      (z+h)|_{t=0}=\eta\\
      u|_{t=0}=0
    \end{array}
  \right. .
\end{equation}
The goal is to find out whether or not the discrete
dynamical system (\ref{fvm_2D_eq_frac.07}) with
(\ref{fvm_2D_eq.05}) and (\ref{eq_visc1}) reaches a steady
state.  If a steady state is reached, then we are interested
to see the way this state is determined by the initial water
level $\eta$.  As in the aforementioned papers, we consider
the cases corresponding to values of the left boundary flux
$q_0=1.53$m$^2$/s and $q_0=0.18$m$^2$/s.

Figures~\ref{fig_validate_shm_shs_case1} and
\ref{fig_validate_shm_shs_case2} illustrate the behavior of
the numerical solution when it develops a moving or a steady
shock.
\begin{figure}[htbp]
  \hspace{-3mm}
  \begin{tabular}{ cccccc }
    {} & \hspace{5mm}\tiny{$t=5$s} & \hspace{5mm}\tiny{$t=20$s} & \hspace{5mm}\tiny{$t=150$s} & \tiny{$t=200$s}& \\
    \begin{turn}{90}\hspace{4mm}\tiny{Water level [m]}\end{turn}\hspace{-5mm}
       & \includegraphics[width=0.24\textwidth]{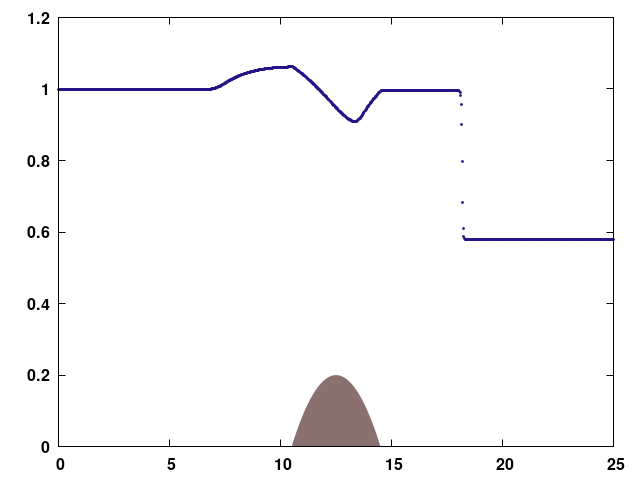}\hspace{-5mm}
       & \includegraphics[width=0.24\textwidth]{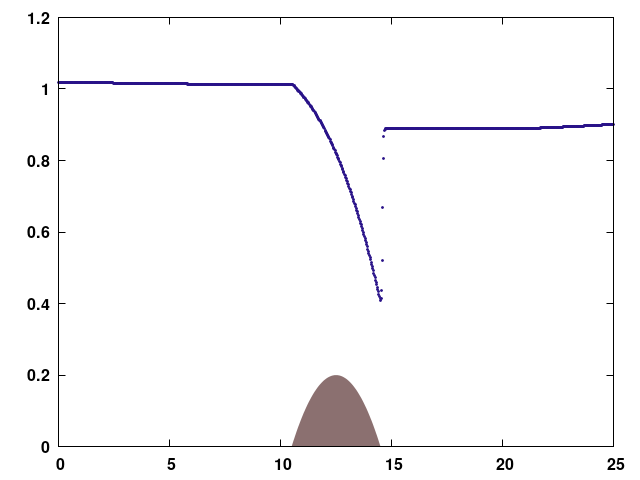}\hspace{-5mm}
       & \includegraphics[width=0.24\textwidth]{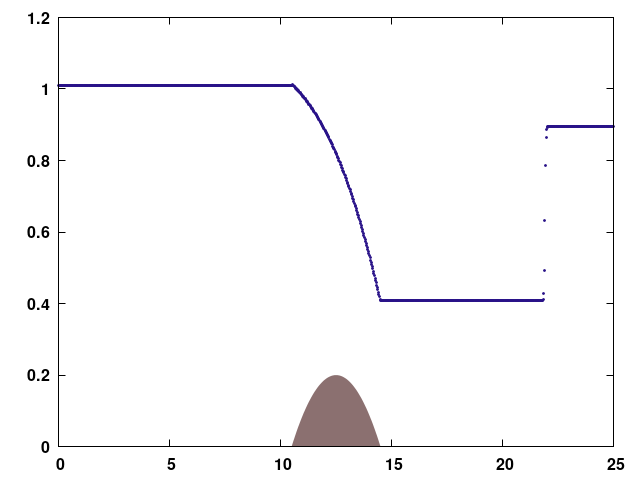}\hspace{-5mm}
       & \includegraphics[width=0.24\textwidth]{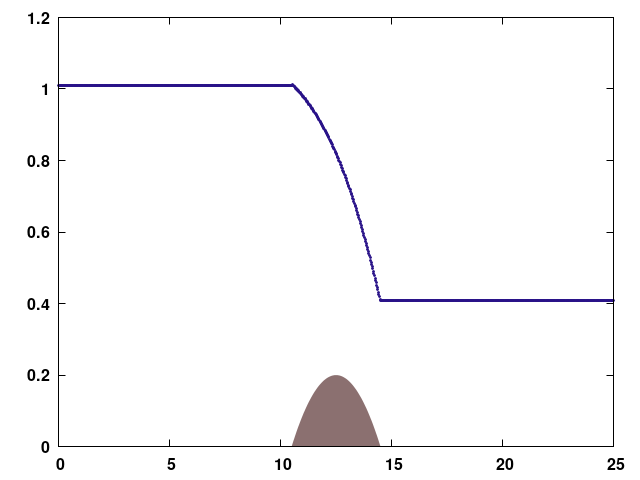}\hspace{-3mm}
       & \begin{turn}{90}\hspace{7mm}\tiny{$\eta=0.58$m}\end{turn}\\
    \begin{turn}{90}\hspace{4mm}\tiny{Water level [m]}\end{turn}\hspace{-5mm}
       & \includegraphics[width=0.24\textwidth]{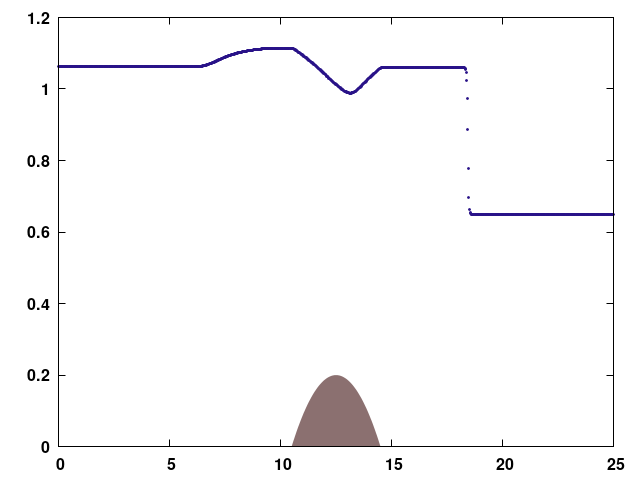}\hspace{-5mm}
       & \includegraphics[width=0.24\textwidth]{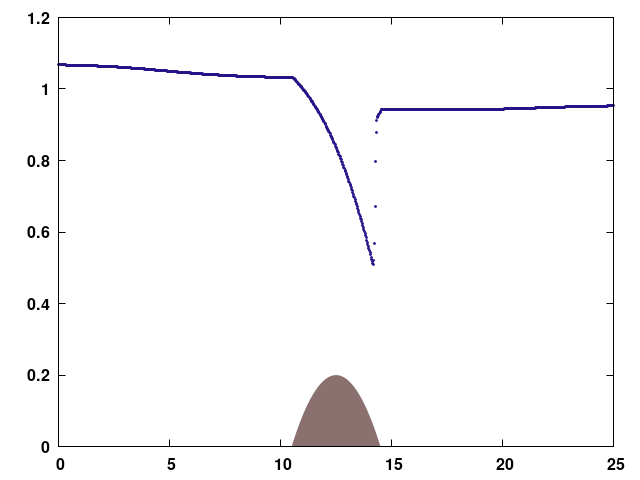}\hspace{-5mm}
       & \includegraphics[width=0.24\textwidth]{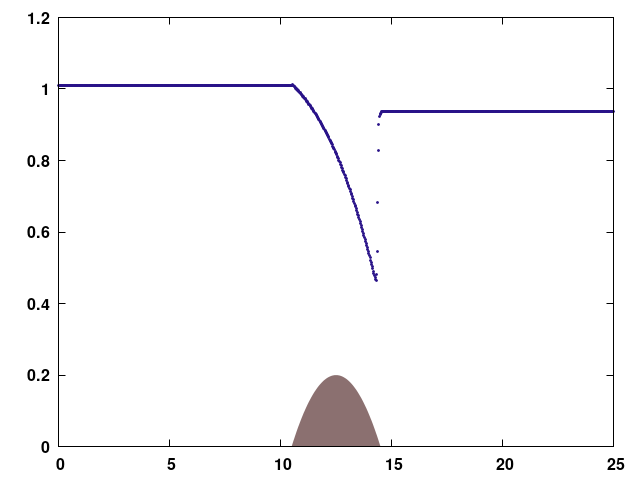}\hspace{-5mm}
       & \includegraphics[width=0.24\textwidth]{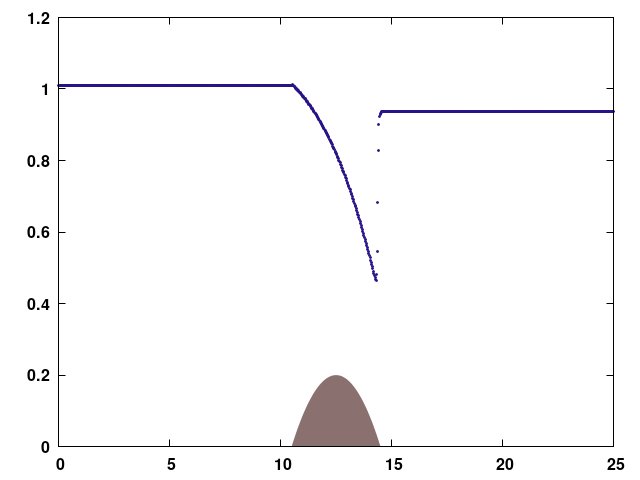}\hspace{-3mm}
       & \begin{turn}{90}\hspace{7mm}\tiny{$\eta=0.65$m}\end{turn}\vspace{-2mm}\\
    {} & \hspace{5mm}\tiny{$x$ [m]} & \hspace{5mm}\tiny{$x$ [m]} & \hspace{5mm}\tiny{$x$ [m]} & \tiny{$x$ [m]}&
  \end{tabular}
  \caption{Snapshots from the steady solution of a flow over
    a bump.  Moving (first row) and steady (second row) shocks
    for $q_0=1.53$m$^2$/s.}
  \label{fig_validate_shm_shs_case1}
\end{figure}

\begin{figure}[htbp]
  \hspace{-3mm}
  \begin{tabular}{ cccccc }
    {} & \hspace{5mm}\tiny{$t=5$s} & \hspace{5mm}\tiny{$t=20$s} & \hspace{5mm}\tiny{$t=150$s} & \tiny{$t=200$s}& \\
    \begin{turn}{90}\hspace{4mm}\tiny{Water level [m]}\end{turn}\hspace{-5mm}
       & \includegraphics[width=0.24\textwidth]{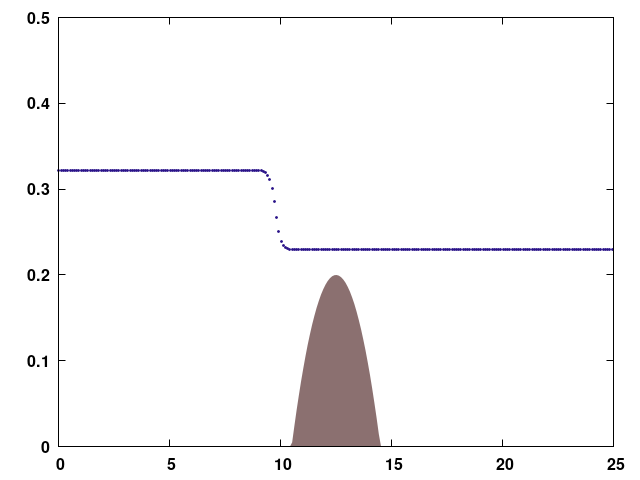}\hspace{-5mm}
       & \includegraphics[width=0.24\textwidth]{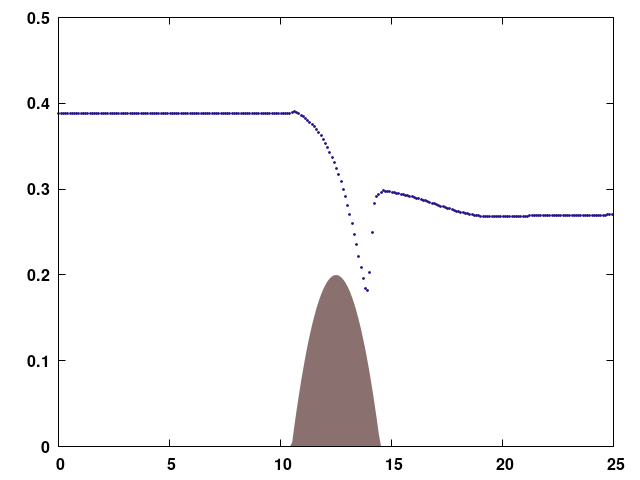}\hspace{-5mm}
       & \includegraphics[width=0.24\textwidth]{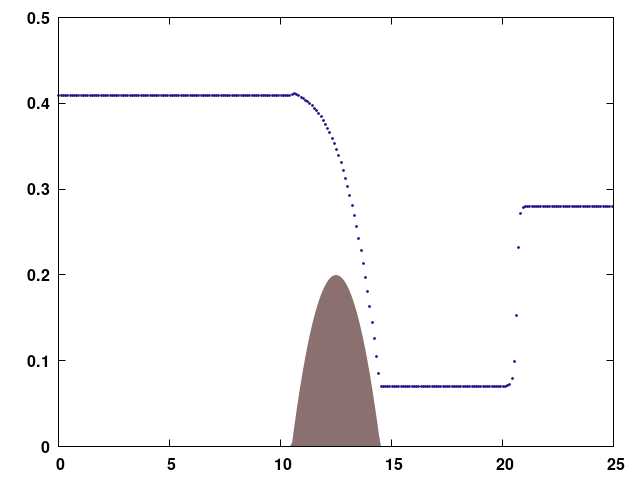}\hspace{-5mm}
       & \includegraphics[width=0.24\textwidth]{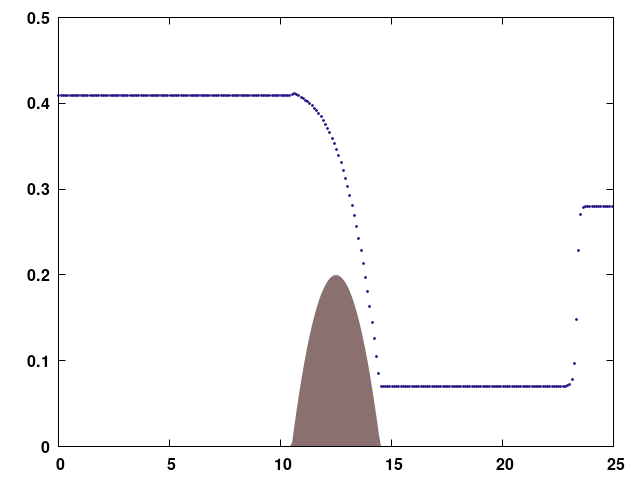}\hspace{-3mm}
       & \begin{turn}{90}\hspace{7mm}\tiny{$\eta=0.23$m}\end{turn}\\
    \begin{turn}{90}\hspace{4mm}\tiny{Water level [m]}\end{turn}\hspace{-5mm}
       & \includegraphics[width=0.24\textwidth]{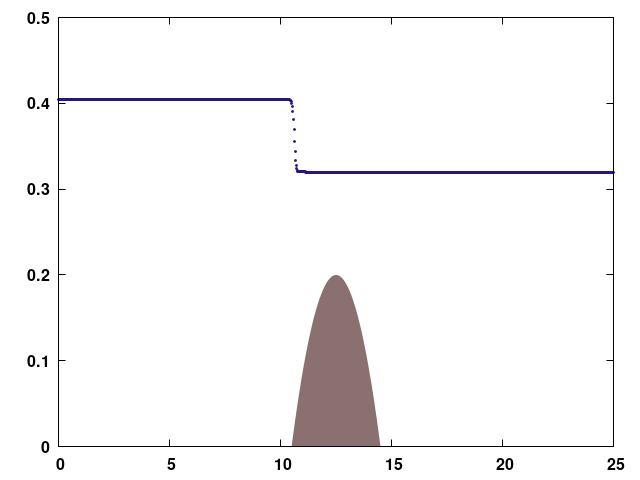}\hspace{-5mm}
       & \includegraphics[width=0.24\textwidth]{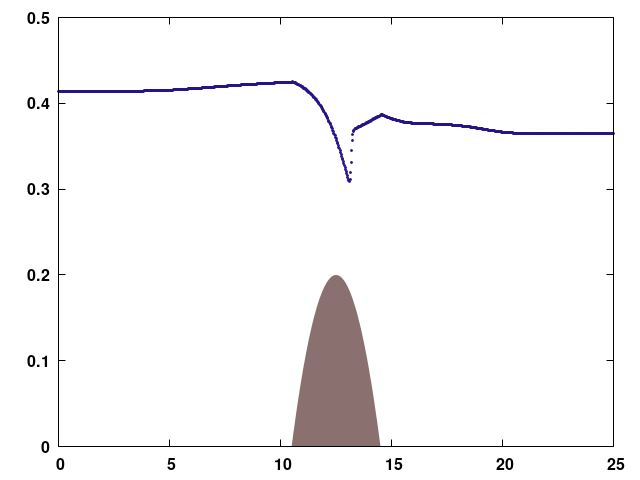}\hspace{-5mm}
       & \includegraphics[width=0.24\textwidth]{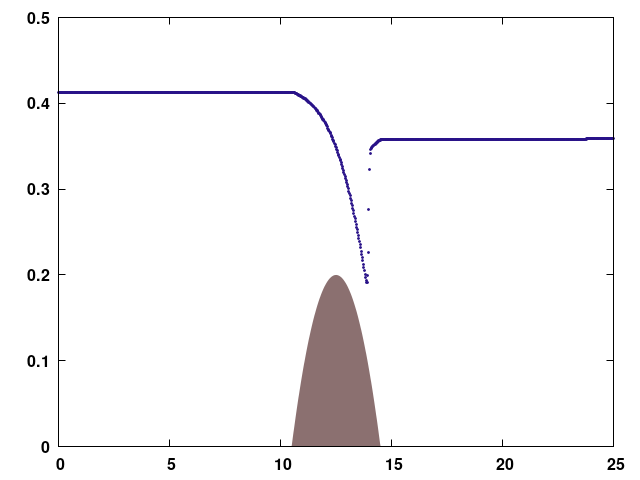}\hspace{-5mm}
       & \includegraphics[width=0.24\textwidth]{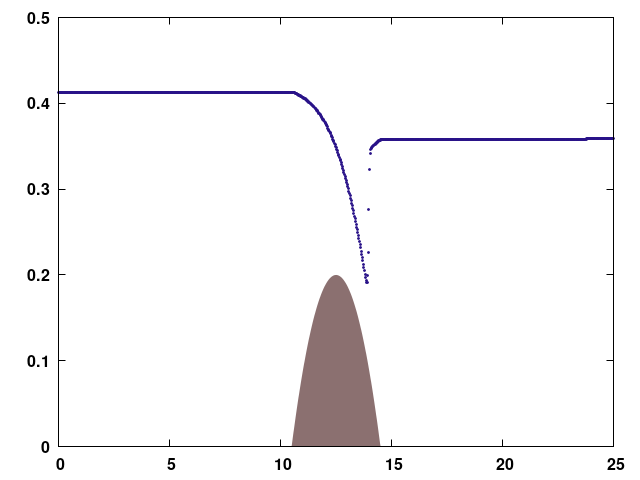}\hspace{-3mm}
       & \begin{turn}{90}\hspace{7mm}\tiny{$\eta=0.32$m}\end{turn}\vspace{-2mm}\\
    {} & \hspace{5mm}\tiny{$x$ [m]} & \hspace{5mm}\tiny{$x$ [m]} & \hspace{5mm}\tiny{$x$ [m]} & \tiny{$x$ [m]}&
  \end{tabular}
  \caption{Snapshots from the steady solution of a flow over
    a bump.  Moving (first row) and steady (second row) shocks
    for $q_0=0.18$m$^2$/s.}
  \label{fig_validate_shm_shs_case2}
\end{figure}

\medskip
\underline{Propagation of a wave}
\medskip
\newline
\indent The second test problem concerns the propagation of a wave
over a bump.  This kind of the test was firstly considered
by LeVeque \cite{leveque}.  The soil topography is defined
by
\begin{equation}
  \label{validation.eq-05}
  z(x) = 
  \left\{
    \begin{array}{ll}
      0.25(\cos{(\pi x/0.1)}+1), & {\rm if}\; |x|<0.1\\
      0, & {\rm otherwise}
    \end{array}
  \right.
\end{equation}
over the domain of flow $\Omega = [-1.25,1.25]$.  The
initial datum is a flat surface perturbed by a small crest
\begin{equation}
  \label{validation.eq-06}
  (z+h)|_{t=0} = 
  \left\{
    \begin{array}{ll}
      1+\epsilon & {\rm if}\; -0.4<x<-0.3\\
      1, & {\rm otherwise}
    \end{array}
  \right. .
\end{equation}
There are no boundary conditions.  The water level on
$\Omega$ at a given moment of time calculated with our
numerical scheme is pictured in
Figure~\ref{fig_validate-wave_over_bump} for
$x\in [-0.5,0.5]$.
\begin{figure}[h!]
  \centering
  \begin{tabular}{ cccc }
    {} & \hspace{5mm}\tiny{$\epsilon=0.2$} & 
    \hspace{5mm}\tiny{$\epsilon=0.2$} & 
    \hspace{5mm}\tiny{$\epsilon=0.01$}\\
    \begin{turn}{90}\hspace{8mm}\tiny{Water level [m]}\end{turn}\hspace{-5mm}
       &\includegraphics[width=0.32\linewidth]{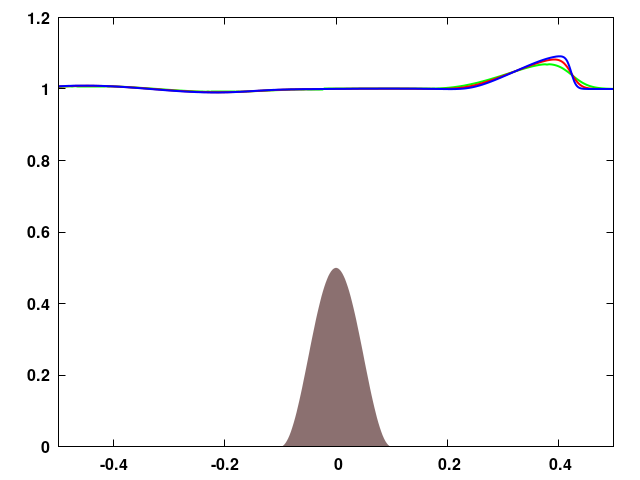}\hspace{-5mm}
       &\includegraphics[width=0.32\linewidth]{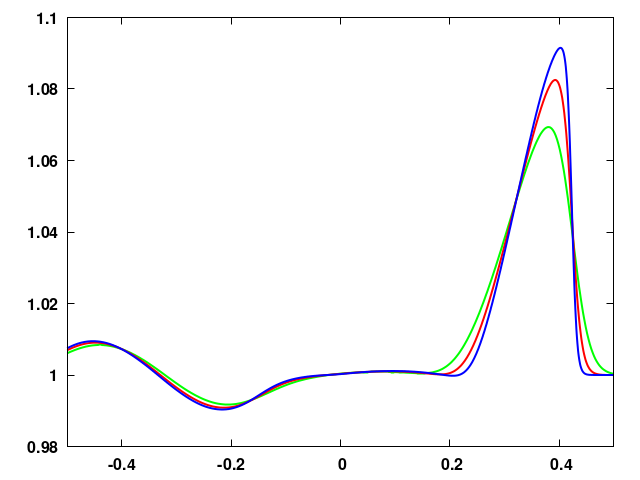}\hspace{-5mm}
       &\includegraphics[width=0.32\linewidth]{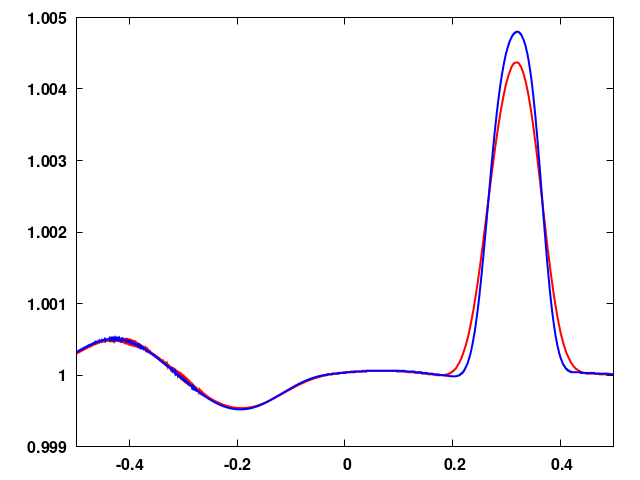}\vspace{-3mm}\\
    {} & \hspace{5mm}\tiny{$x$ [m]} & \hspace{5mm}\tiny{$x$ [m]} & \hspace{5mm}\tiny{$x$ [m]}
  \end{tabular}
  \caption{Snapshot at $t=0.7$s from the propagation of a
    wave over a bump.  This test problem was accomplished
    for the gravitational acceleration $g=1$.  Left picture
    presents the soil topography and the water level given
    by the numerical solution with $N$ interior points:
    $1000$ - green, $2000$ - red, $4000$ - blue.  The
    picture in the middle represents a zoom of the left
    image.  A similar zoom for the numerical solution (with
    $2000$ and $4000$ interior points) for a smaller height
    of the crest is drawn in the right picture.}
  \label{fig_validate-wave_over_bump}
\end{figure}

\subsubsection{Thacker Problem}
The oscillation of the water level in a convex soil surface
or the oscillatory motion of the sea water in the shoreline
are very natural phenomena.  A very difficult problem here
is that a singularity appears at the wet/dry contact between
water and soil surfaces.  Different numerical schemes were
proposed to solve this kind of problem and among them we
mention \cite{cell_dray}.

Thacker problem is a useful test to check the ability of a
numerical scheme to catch and highlight the propagation of a
wet/dry front.  Moreover, this problem can be exactly solved
\cite{thacker, sampson}.  Thus, once can take advantage of
its analytic solution to verify the accuracy of the
numerical scheme.  Thacker solution for 1-D shallow water
equations was obtained in the absence of vegetation
($\theta = 1$), for $\mathfrak{t}^{p} = \boldsymbol{0}$ and
$\mathfrak{t}^{s} = h \tau \boldsymbol{v}$ in
(\ref{swe_vegm_rm.02}), where $\tau$ is a proportionality
coefficient, for a soil surface $z(x)$ of parabolic shape,
and for a water velocity which does not depend on space
variable, i.e. $\boldsymbol{v}(t,x)=\boldsymbol{f}(t)$.

For the case of a 1D flow, the general solution is given by
\begin{equation}
  \begin{split}
    h(t,x)+z(x) = & \left(\displaystyle\frac{x}{g}-
      \displaystyle\frac{\psi(t)}{2g\lambda_1\lambda_2}+
      \displaystyle\frac{\alpha}{\lambda_1\lambda_2}\right)\psi(t)+\eta_0,\\
    u(t) = & A\exp{\lambda_1 t}+B\exp{\lambda_2 t},
  \end{split}
  \label{eq_sol_thacker}
\end{equation}
where
\begin{equation*}
  \psi(t) = A\lambda_2\exp{\lambda_1 t}+B\lambda_1\exp{\lambda_2 t},
\end{equation*}
$\lambda_1$ and $\lambda_2$ solve
\begin{equation*}
  \left\{
  \begin{split}
    \lambda_1+\lambda_2 = & -\tau\\
    \lambda_1\lambda_2 = & g\partial^2_xz(x)
  \end{split}
  \right.
\end{equation*}
and
\begin{equation*}
  \alpha = \partial_xz(x)-x\partial^2_xz(x).
\end{equation*}
In (\ref{eq_sol_thacker}), $A$ and $B$ are integration
constants.  Depending on $z(x)$ and $\tau$, the solutions
$\lambda_1$ and $\lambda_2$ can be real or complex numbers.
An oscillating solution is obtained if
\begin{equation*}
  \tau^2-4g\partial^2_xz(x)<0.
\end{equation*}
In our numerical experiment, we chose
\begin{equation*}
  \begin{array}{l}
    z(x)=z_m\left(\displaystyle\frac{2x}{L}-1\right)^2,\quad x\in [0,L],\\
    A+B=0; \quad A-B={\rm i}u_0,
  \end{array}
\end{equation*}
where the constants $z_m$, $L$ and $u_0$ are free
parameters.

We compare the results from our numerical scheme (using
different formulas (\ref{fvm_2D_eq.05}) or
(\ref{fvm_2D_eq.05_varianta3}) for the interface values of
$(\theta h)^s$) with the analytic Thacker solution: the
evolution of the errors between exact and numerical
solutions in Figure \ref{fig_th_1} and the dynamics of the
water surface in Figure \ref{fig_th_2}.  In all simulations,
we use $z_m=10 {\rm m}$, $L=6000 {\rm m}$,
$u_0=5 {\rm m}\cdot {\rm s}^{-1}$ and
$\tau=0.001 {\rm s}^{-1}$. The artificial viscosity is
defined by (\ref{eq_visc1}).

\begin{figure}[h!]
  \centering
  \begin{tabular}{ cccc }
    {} & \hspace{5mm}\tiny{$\delta x = 10$m} & 
    \hspace{5mm}\tiny{$\delta x = 5$m} & 
    \hspace{5mm}\tiny{$\delta x = 1$m}\\
    \begin{turn}{90}\hspace{14mm}\tiny{Error}\end{turn}\hspace{-5mm}
       &\includegraphics[width=0.32\linewidth]{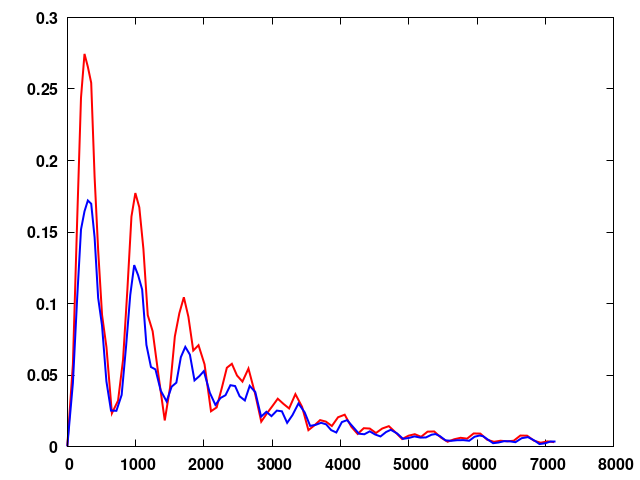}\hspace{-4mm}
       &\includegraphics[width=0.32\linewidth]{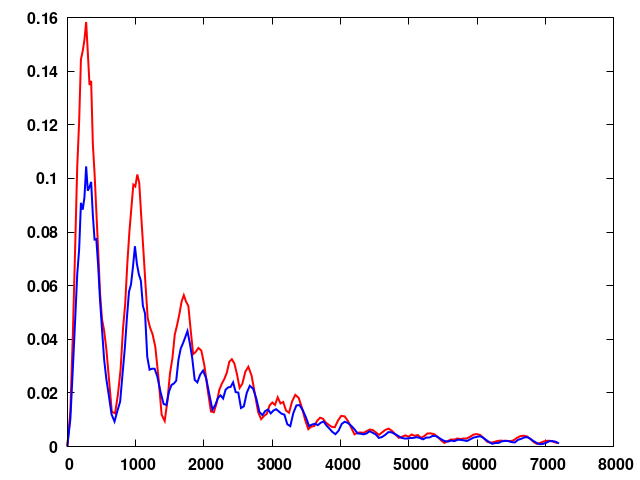}\hspace{-4mm}
       &\includegraphics[width=0.32\linewidth]{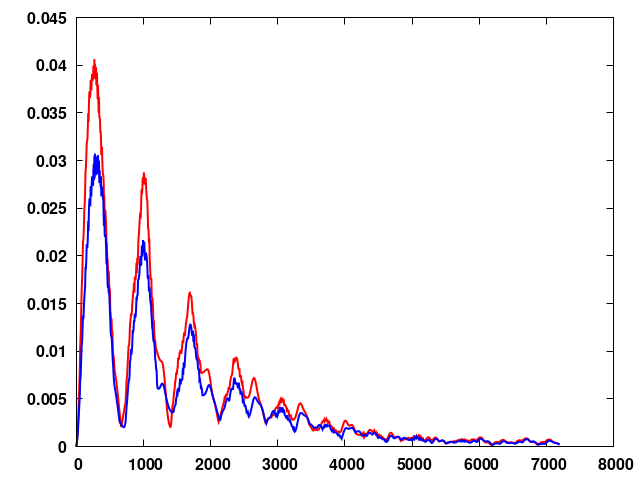}\vspace{-3mm}\\
    {} & \hspace{5mm}\tiny{time [s]} & \hspace{5mm}\tiny{time [s]} & \hspace{5mm}\tiny{time [s]}
  \end{tabular}
  \caption{Thacker Problem: time evolution of the error
    $err=||(hu)^{app}-(hu)^{ex}||_{\infty}$ of discharge
    between the analytic and numerical solutions, when the
    numerical one was calculated with the arithmetic mean
    formula (\ref{fvm_2D_eq.05_varianta3}), blue line, and
    the upwind formula (\ref{fvm_2D_eq.05}), red line, for
    $(\theta h)^s$, respectively.  This evolution of the
    error is represented for three space step-sizes
    $\delta x$ and one can observe that the error is
    decreasing with $\delta x$.}
  \label{fig_th_1}
\end{figure}

\begin{figure}[h!]
  \centering
  \begin{tabular}{ ccccc }
    {} & \hspace{5mm}\tiny{$t = 300$s, $err= 0.0435$} & 
    \hspace{5mm}\tiny{$t = 600$s, $err= 0.0513$} & 
    \hspace{5mm}\tiny{$t = 900$s, $err= 0.0460$} & 
    \tiny{$t = 1200$s, $err= 0.0348$} \\
    \begin{turn}{90}\hspace{4mm}\tiny{Water level [m]}\end{turn}\hspace{-5mm}
       &\includegraphics[width=0.24\linewidth]{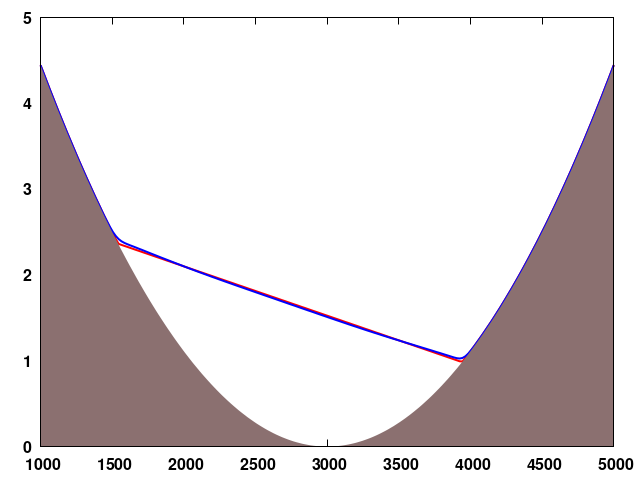}\hspace{-5mm}
       &\includegraphics[width=0.24\linewidth]{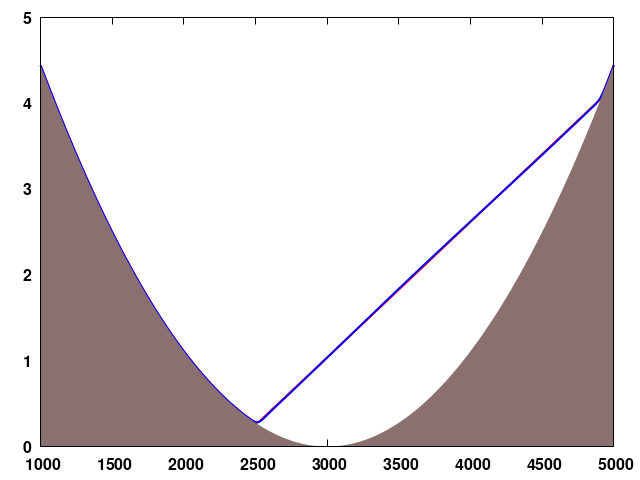}\hspace{-5mm}
       &\includegraphics[width=0.24\linewidth]{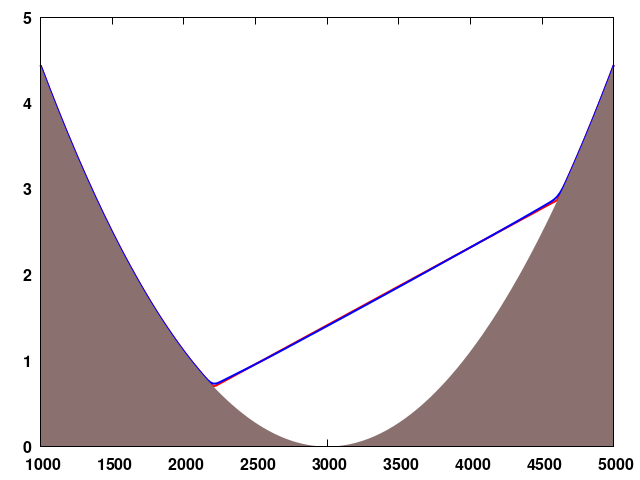}\hspace{-5mm}
       &\includegraphics[width=0.24\linewidth]{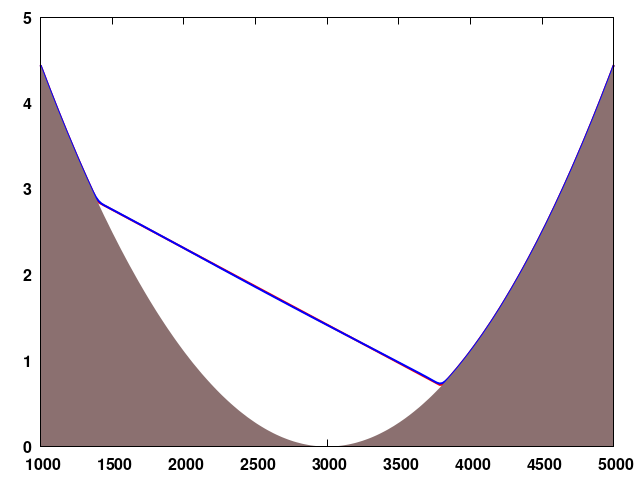}\vspace{-3mm}\\
    {} & \hspace{5mm}\tiny{$x$ [m]} & \hspace{5mm}\tiny{$x$ [m]} & \hspace{5mm}\tiny{$x$ [m]} & \tiny{$x$ [m]}
  \end{tabular}
  \caption{Thacker Problem: dynamics of the water level.
    Numerical (blue) and analytic (red) solutions.  The
    error between these two solutions is given by
    $err=||((h+z)^{app}-(h+z)^{ex})/(h+z)^{ex}||_{\infty}$.}
  \label{fig_th_2}
\end{figure}

\subsubsection{Riemann Problem}
\label{sect_RiemannProblem}
We remind that Riemann problem for shallow water equations
with topography and vegetation requires to solve the
equations for piecewise constant initial data
\[
\left.(h,u)\right|_{t=0}=\left\{
\begin{array}{ll}
  (h_L,u_L),&x<0\\
  (h_R,u_R),&x>0
\end{array}
\right.
\]  
and piecewise constant soil surface and porosity functions
\[
(\theta,z)(x)=\left\{
\begin{array}{ll}
  (\theta_L,z_L),&x<0\\
  (\theta_R,z_R),&x>0
\end{array}
\right. .
\]
The problem was extensively analyzed \cite{LeFloch2007,
  LEFLOCH20117631, toro, han, alcrudo, adrianov} for the
case of constant vegetation, but there are only few results
when the porosity is a variable function
\cite{GUINOT2017133, COZZOLINO201883}.  Perhaps, this
scarcity of results is due to the fact that the problem is
less tractable for the case of jumps in both soil and
porosity functions.

Usually, the solution of the problem is sought as a bunch of
simple waves (shock waves, rarefaction waves and standing
waves) ordered on the base of the entropy principle,
\cite{lax}.  We remark that such a solution may not exist or
it can exist (being unique or non-unique).  We perform a
comparative study on these three cases. The data for each
case are given in Table~\ref{tabel_RP}.
\begin{table}[h!]
  \centering
  \begin{tabular}{|c|c|c|c|}
    \hline
 &{Non-existence}&{Uniqueness}&{Two solutions}\\\hline
    $(h,u)_L$&\multicolumn{3}{c|}{$(0.2,5.0)$}\\\cline{2-4}
    $(h,u)_R$&$(0.6,0.4)$&$(0.6,1.34)$&$(0.6,0.4)$\\\hline
    $(\theta,z)_L$&$(0.3,1.0)$& \multicolumn{2}{c|}{$(0.8,1.0)$}\\\cline{2-4}
    $(\theta,z)_R$&\multicolumn{3}{c|}{$(1.0,1.2)$}\\\hline    
  \end{tabular}
  \caption{Data for Riemann problem}
  \label{tabel_RP}
\end{table}
The choice of these data was suggested by some cases
analyzed in \cite{LEFLOCH20117631}.

\begin{figure}[htbp]
  \centering
  \begin{tabular}{m{0.001\linewidth}m{0.22\linewidth}m{0.22\linewidth}m{0.22\linewidth}m{0.22\linewidth}}
    {}
    &{\begin{center}Non-existence\end{center}}
    &{\begin{center}Uniqueness\end{center}}
    &\multicolumn{2}{c}{Two solutions}\\
    {}
    &\tiny{$z_R - z_L = 0.2, \, \frac{\theta_L}{\theta_R} = 0.3$}
    &\tiny{$z_R - z_L = 0.2, \, \frac{\theta_L}{\theta_R} = 0.8$}
    &\multicolumn{2}{c}{$z_R - z_L = 0.2, \, \frac{\theta_L}{\theta_R} = 0.8$}\\
    \begin{turn}{90}\hspace{1mm}\tiny{u [m/s]}\end{turn}
    &\includegraphics[width=\linewidth]{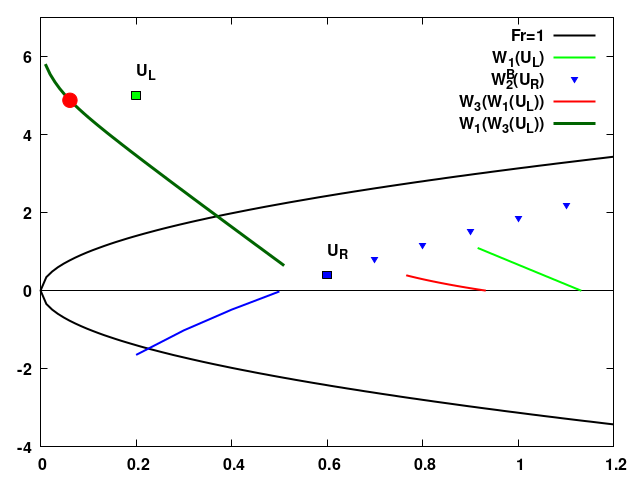}
    &\includegraphics[width=\linewidth]{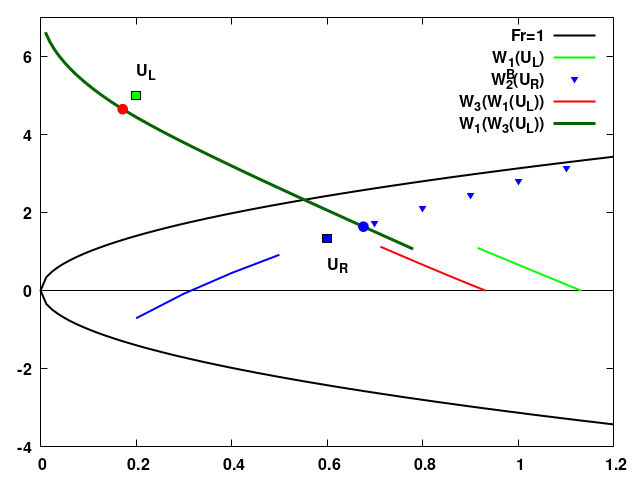}
    &\includegraphics[width=\linewidth]{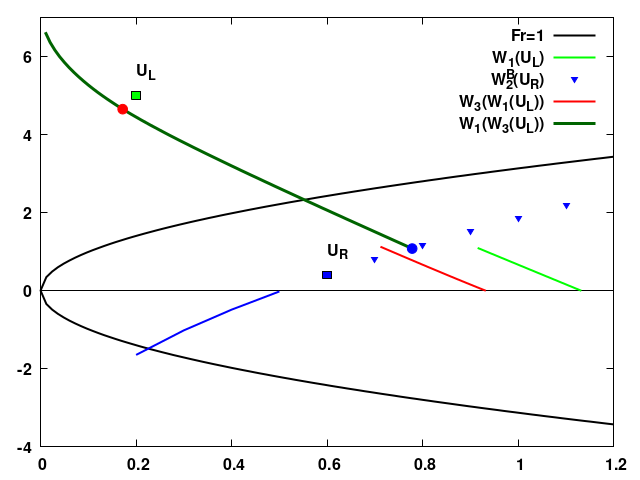}
    &\includegraphics[width=\linewidth]{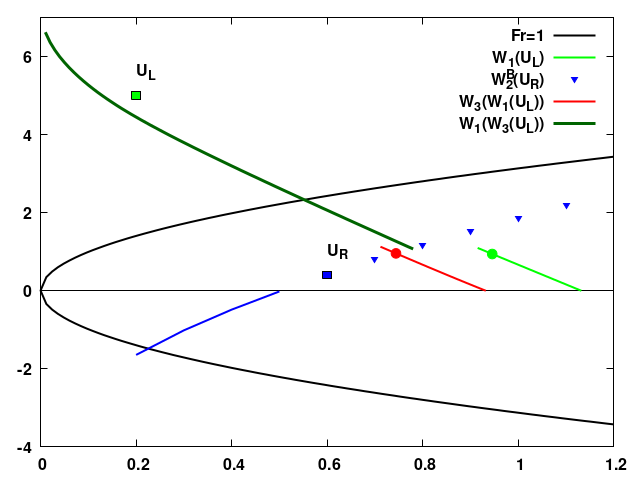}\\
    & \hspace{13mm}\tiny{h [m]}
    & \hspace{13mm}\tiny{h [m]}
    & \hspace{13mm}\tiny{h [m]}
    & \hspace{13mm}\tiny{h [m]}\\
    &&&&\\
    \begin{turn}{90}\tiny{Water level [m]}\end{turn}
    &\includegraphics[width=\linewidth]{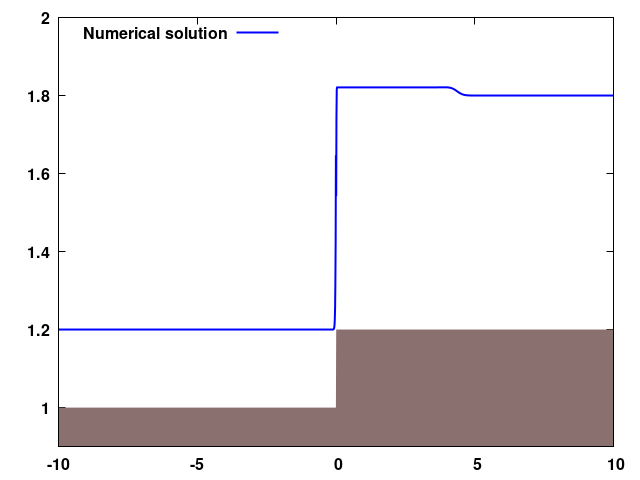}
    &\includegraphics[width=\linewidth]{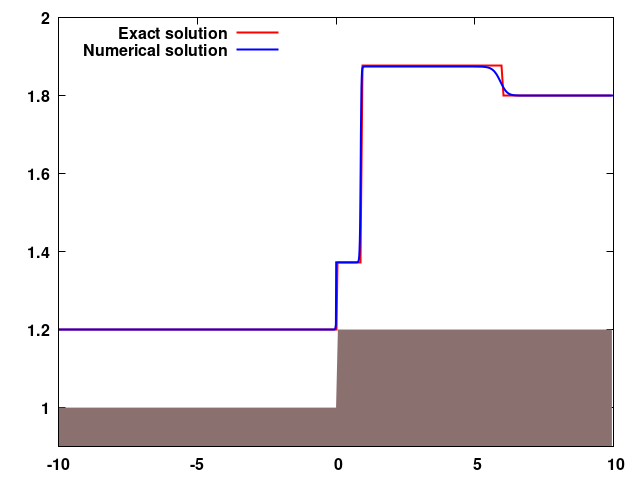}
    &\hspace{2mm}\begin{tabular}{c} Solution \\ not revealed \end{tabular}
    &\includegraphics[width=\linewidth]{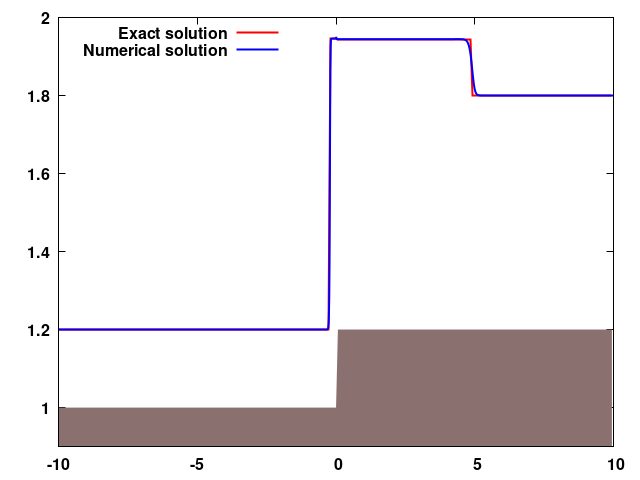}\\
    \begin{turn}{90}\hspace{1mm}\tiny{Water level [m]}\end{turn}
    &\includegraphics[width=\linewidth]{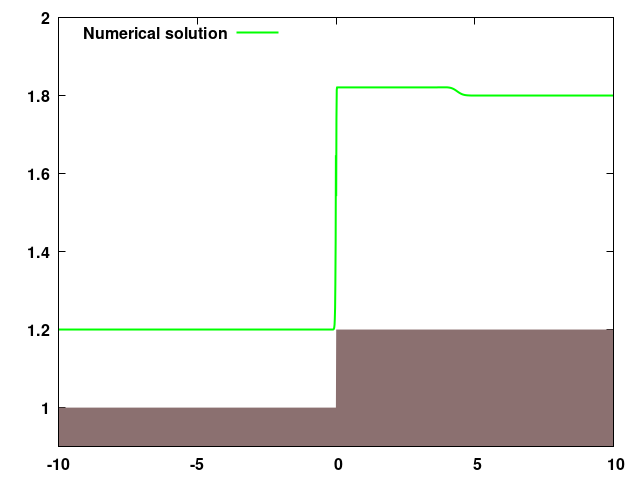}
    &\includegraphics[width=\linewidth]{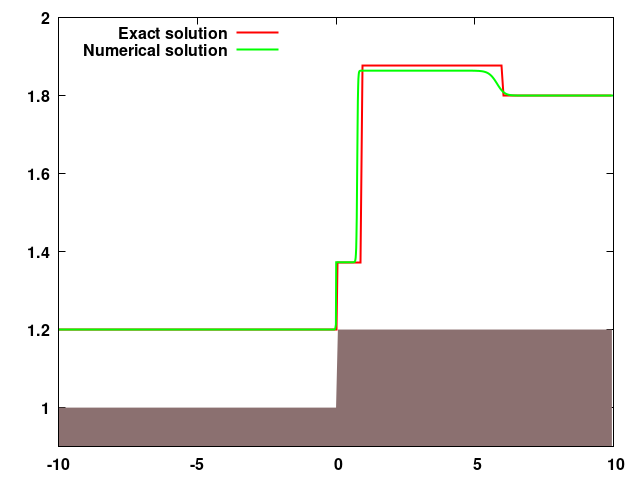}
    &\includegraphics[width=\linewidth]{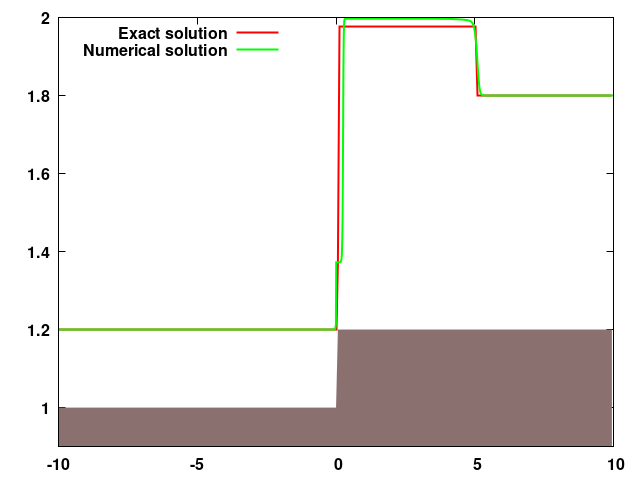}
    &\hspace{2mm}\begin{tabular}{c}Solution \\ not revealed \end{tabular}\\		       
    {}
    & \hspace{13mm}\tiny{x [m]}
    & \hspace{13mm}\tiny{x [m]}
    & \hspace{13mm}\tiny{x [m]}
    & \hspace{13mm}\tiny{x [m]}
  \end{tabular}
  \caption{Riemann problem for three different initial data.
    The first row of pictures contains the simple wave
    curves in the phase space $(h,u)$.  The last two rows of
    pictures include the graphics of the numerical and exact
    solutions (if they exist) of the free water surface at
    $t=1.5$s.  The numerical solutions represented in the
    second and third rows are calculated using
    formulas~(\ref{fvm_2D_eq.05_varianta3}) and
    (\ref{fvm_2D_eq.05}), respectively.}
  \label{RP_graph}
\end{figure}
The analytic solutions for the Riemann problem are
calculated using the path connection introduced in
\cite{sds-riemann-constanta} and the algorithms described in
\cite{smc-rp-19}.

Figure~\ref{RP_graph} summarizes the results for the Riemann
problem in the three cases (non-existence, uniqueness and
two solutions) associated to the initial data given in
Table~\ref{tabel_RP}.
\bigskip

\subsection{External Validation}
\label{sect_ExternalValidation}
Numerical data incorporate two main errors: one given by the
approximation of the physical phenomena by a mathematical
model and the other one given by the approximation of the
solution of the mathematical model.  Consequently, a good
fit of the numerical data with experimental data validates
both the mathematical model and the numerical scheme.  Thus,
for real world applications, it is essential to confront the
numerical data with the experimental data.

We test the mathematical model and the numerical solution
considering the following problems:
\begin{enumerate}
  \item dam break flow over a triangular bump,
  \item downslope flow through rigid vegetation,
  \item flow on a convex-concave vegetated soil surface,
  \item simulation on Paul's Valley (2D).
\end{enumerate}
The first two tests simulate solutions for experiments with
measured data, while the last two tests target the
qualitative behavior of the solution.

\subsubsection{Dam Break Flow Over a Triangular Bump}
First, we consider the CADAM test case of the dam break flow
propagation in an impermeable channel with a symmetric
triangular bump \cite{cadam, ERPICUM20102143, Hsu2019}.  The
layout of this experiment and the initial state of the water
at rest are presented in Figure~\ref{fig_sch_tr_bump}.
\begin{figure}[!htbp]
  \includegraphics[width=0.99\textwidth]{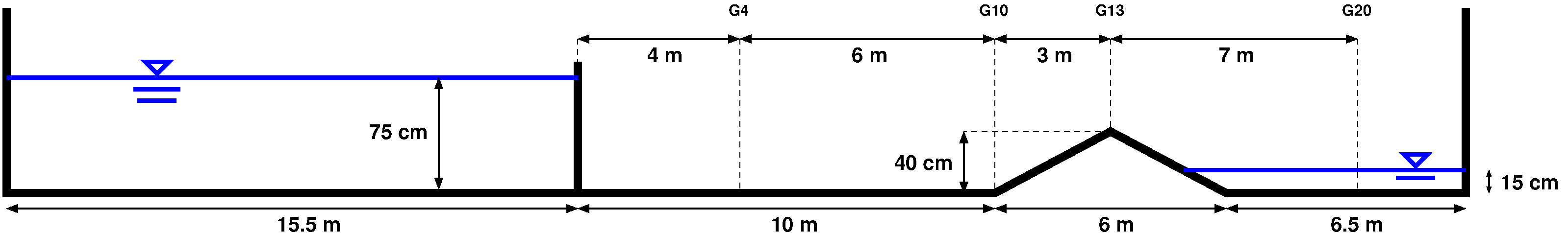}
  \caption{Scheme of the experimental installation for the
    triangular bump test}
  \label{fig_sch_tr_bump}
\end{figure}
Several gauges $G_i$ are placed in the channel to measure
the water depth evolution once the gate retaining the water
volume from the left side is suddenly removed.

In this experiment, the cover plant is absent: $\theta=1$.
For the frictional term
$\mathfrak{t}^{s}_a=\alpha_s(h)|\boldsymbol{v}|v_a$, we
consider the Darcy-Weisbach formula
\begin{equation}
  \alpha_s(h) = \tau,
  \label{eq_Darcy-Weisbach}
\end{equation}
where $\tau$ is constant.  For this setup, we performed
several numerical simulations with different values for
$\tau$.  The experimental data are extracted from the
graphics reported in \cite{Hsu2019}, Figure~5.  The results
are synthesized in Figure~\ref{fig_flow_tr_bump}.
\begin{figure}[!htbp]
  \centering
  \begin{tabular}{ ccccc }
    {} & \hspace{5mm}\tiny{G4} & \hspace{5mm}\tiny{G10} & \hspace{5mm}\tiny{G13} & \tiny{G20} \\
    \begin{turn}{90}\hspace{4mm}\tiny{Water depth [m]}\end{turn}\hspace{-5mm}
       & \includegraphics[width=0.24\textwidth]{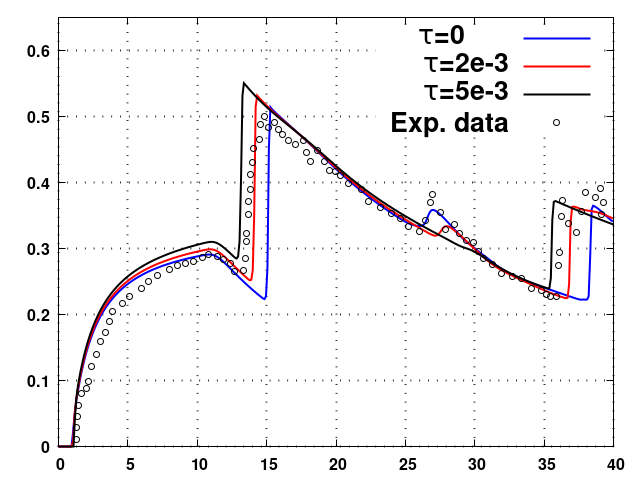}\hspace{-5mm}
       & \includegraphics[width=0.24\textwidth]{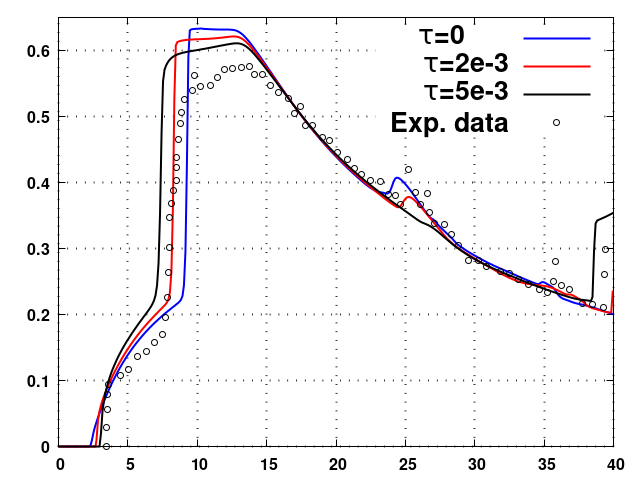}\hspace{-5mm}
       & \includegraphics[width=0.24\textwidth]{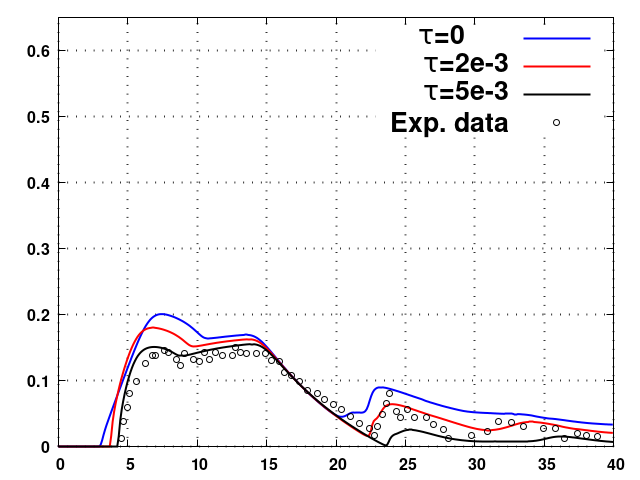}\hspace{-5mm}
       & \includegraphics[width=0.24\textwidth]{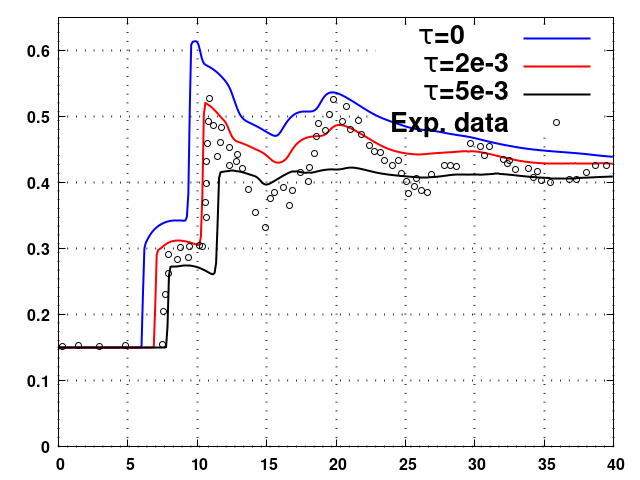}\vspace{-3mm}\\
    {} & \hspace{5mm}\tiny{Time [s]} &\hspace{5mm}\tiny{Time [s]} & \hspace{5mm}\tiny{Time [s]} & \tiny{Time [s]}
  \end{tabular}
  \caption{Flow over a triangular bump}
  \label{fig_flow_tr_bump}
\end{figure}
Note that our solution is in agreement to the measured one
(for all four gauges) for $\tau \approx 2e-3$.  

We repeated the numerical simulations using Stickler-Manning
formula
\begin{equation}
  \alpha_s(h)=\frac{n^2}{h^{1/3}}
  \label{Stickler-Manning}
\end{equation}
with $n^2 \approx 0.003$m$^{1/3}$ and we have seen that
there is no significant difference from the previous case.

\subsubsection{Flow Over a Slope with Vegetation}
The flow in the presence of vegetation is more complex than
the one on bare soil and perhaps this explains a scarcity of
the reported data.

Here, we compare our results with the ones reported in
\cite{Dupuis2016}.  Briefly, the experimental installation
consists of an $18$m long and $1$m width laboratory flume
with a longitudinal bottom slope $S=1.05{\rm mm}/{\rm m}$.
The vegetation is modelled using uniformly distributed
emergent circular cylinders of radius $R_c=5$mm.
The density of vegetation is $N=81$ cylinder/m$^2$.  This
type of vegetation allows the estimation of the porosity by
\[
\theta =1-N\pi R_c^2= 0.99336.
\]
To obtain the parameters $\alpha_p$ and $\alpha_s(h)$, we
use the steady state values of $h$ and $u$ reported
\cite{Dupuis2016} for different values of the water
discharge $q = hu$.  These values can be theoretically
estimated from relations
\[
\begin{array}{l}
  \theta h u= q,\\
  g\theta h S=(\alpha_p h (1-\theta)+\theta \alpha_s(h))u^2.
\end{array}
\]
One can estimate $\alpha_s(h)$ from the data for bare soil
and then use it to calculate $\alpha_p$ using the data for
uniformly distributed vegetation.  Thus, if Darcy-Weisbach
formula is used, one finds that $\tau$ runs in
$[6.508084, 9.00375]\cdot 10^{-3}$ and $\alpha_p$ takes
values in $[73.20, 80.20]$.  For the case of
Strickler-Manning formula, one finds
$n^2\in [2.893622, 3.178546]\cdot 10^{-3}$ and
$\alpha_p\in [72.27832, 77.97572]$.

These estimations indicate that Strickler-Manning formula is
more appropriate than Darcy-Weisbach formula to be used for
this kind of flow.

For testing numerical data, we choose the experiment of
water flow on a partial vegetated bottom surface.  We use
Strickler-Manning formula with the parameters
\[
n^2=3.005785\cdot 10^{-3}{\rm m}^{1/3},\quad
\alpha_p=74.7643{\rm m}^{-1}.
\]
These values correspond to the Strickler $K=57.099$ and drag
coefficient $C_D=1.1738$.

Figure~\ref{fig_flow_slope} includes the numerical and
experimental data for six different steady configurations.
The experimental data are extracted from the graphics
reported in \cite{Dupuis2016}, Figure~8.
\begin{figure}[h!]
  \centering
  \begin{tabular}{ cccc }
    {} & \hspace{5mm}\tiny{$Q = 7$ l/s} & 
    \hspace{5mm}\tiny{$Q = 15$ l/s} & 
    \hspace{5mm}\tiny{$Q = 21$ l/s}\\
    \begin{turn}{90}\hspace{8mm}\tiny{Water level [mm]}\end{turn}\hspace{-5mm}
       &\includegraphics[width=0.32\linewidth]{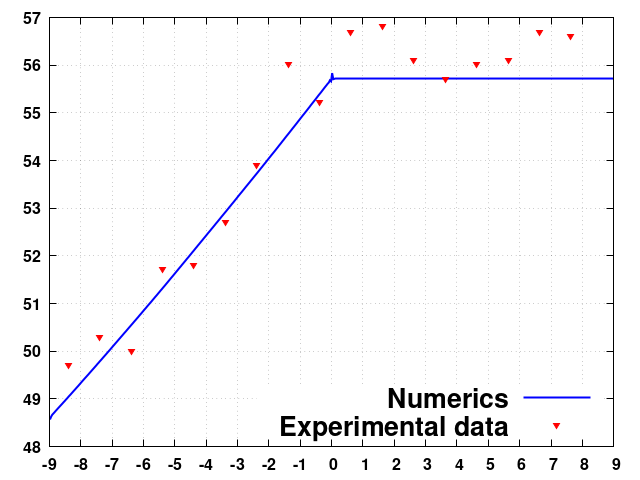}\hspace{-5mm}
       &\includegraphics[width=0.32\linewidth]{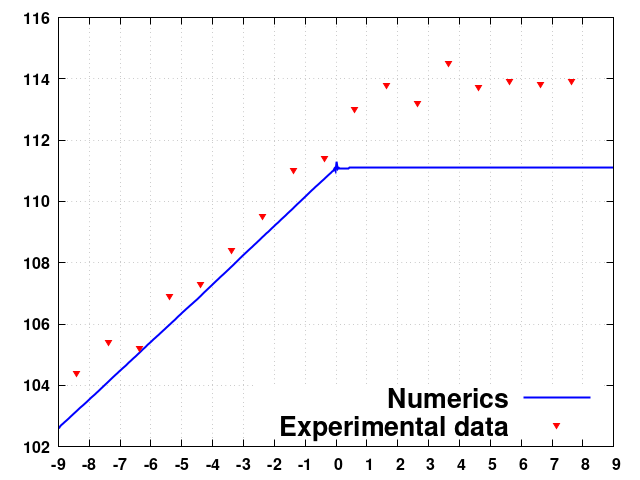}\hspace{-5mm}
       &\includegraphics[width=0.32\linewidth]{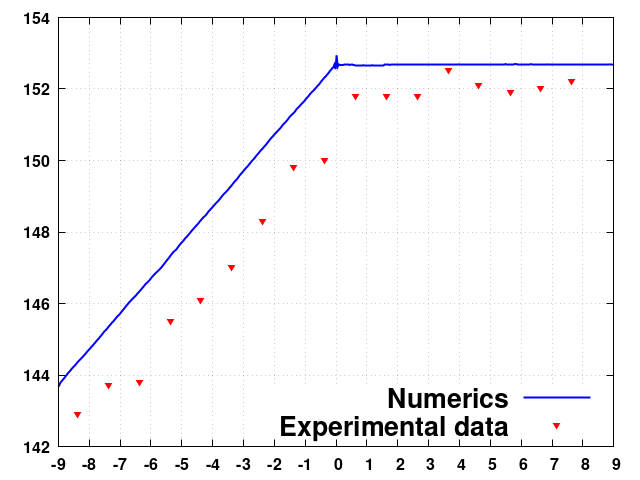}\vspace{-3mm}\\
    {} & \hspace{5mm}\tiny{x [m]} & \hspace{5mm}\tiny{x [m]} & \hspace{5mm}\tiny{x [m]}\\
    {} & \hspace{5mm}\tiny{$Q = 7$ l/s} & 
    \hspace{5mm}\tiny{$Q = 15$ l/s} & 
    \hspace{5mm}\tiny{$Q = 50$ l/s}\\
    \begin{turn}{90}\hspace{8mm}\tiny{Water level [mm]}\end{turn}\hspace{-5mm}
       &\includegraphics[width=0.32\linewidth]{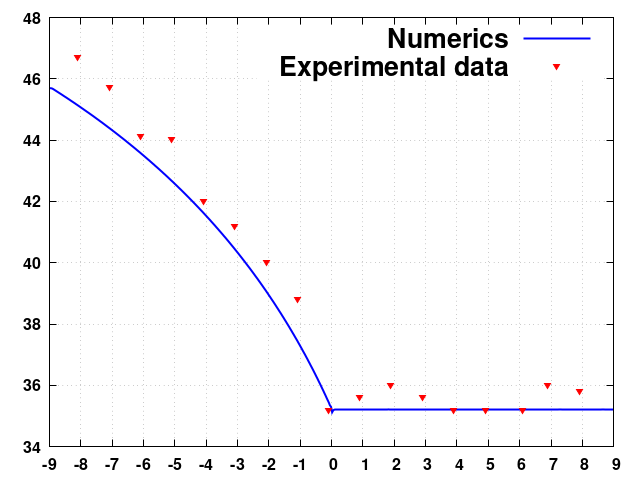}\hspace{-5mm}
       &\includegraphics[width=0.32\linewidth]{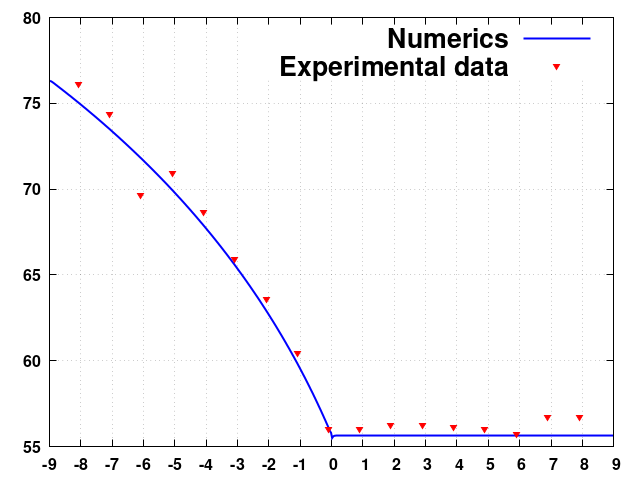}\hspace{-5mm}
       &\includegraphics[width=0.32\linewidth]{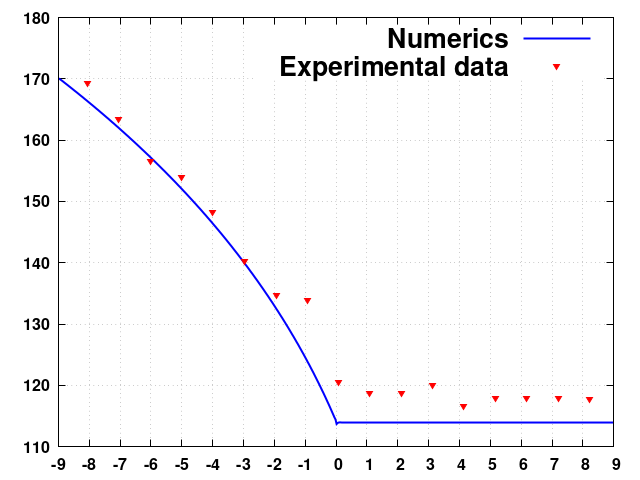}\vspace{-3mm}\\
    {} & \hspace{5mm}\tiny{x [m]} & \hspace{5mm}\tiny{x [m]} & \hspace{5mm}\tiny{x [m]}
  \end{tabular}
  \caption{Flow over a slope with vegetation}
  \label{fig_flow_slope}
\end{figure}

\subsubsection{Flow Over a Convex-Concave Vegetated Soil
  Surface}
\label{sect_FlowOverConvex}
There are two reasons for considering the flow over a
convex-concave vegetated soil surface: to highlight the
influence of the cover plant on the water flow dynamics and
to emphasize the ability of our numerical scheme not only to
preserve the lake as an equilibrium point, but to
asymptotically reach it.

We consider a topography with the surface given by a fourth
degree polynomial on $[-L,L]$
\begin{equation*}
  z(x) = h_m 
  \left[ 
    \sqrt{\frac{h_b}{h_m}} - \frac{1}{L^2}
    \left( 1 + \sqrt{\frac{h_b}{h_m}} \right)
    x^2
  \right]^2,
\end{equation*}
where $L=1500$m, $h_m=50$m and $h_b=5$m.  For this
configuration, the mathematical equilibrium point may
consist of one or two lakes.  Our model is run for two
different configurations, without vegetation and with
vegetation, respectively.  We use Stickler-Manning formula
(\ref{Stickler-Manning}) with $n^2 = 0.003{\rm m}^{1/3}$.
The initial conditions
\begin{equation*}
  h_0 = 
  \left\{
    \begin{array}{ll}
      2, & {\rm if}\; x \leq - L/3\\
      0, & {\rm if}\; x > - L/3
    \end{array}
  \right.
  , \quad v_0 = 0
\end{equation*}
are identical for both configurations.  For the
configuration with vegetation, we add
\begin{equation*}
  \alpha_p=79.3, \quad \theta = 
  \left\{
    \begin{array}{ll}
      0.9, & {\rm if}\; -L/3 \leq x \leq 0\\
      1, & {\rm otherwise}
    \end{array}
  \right. .
\end{equation*}

\begin{figure}[htbp]
  \centering
  \begin{tabular}{ ccc }
    &\small{Bare soil} & \small{Vegetated soil}\\
    \begin{turn}{90}\hspace{10mm}\tiny{Water level
        [m]}\end{turn}
       &\includegraphics[width=0.35\linewidth]{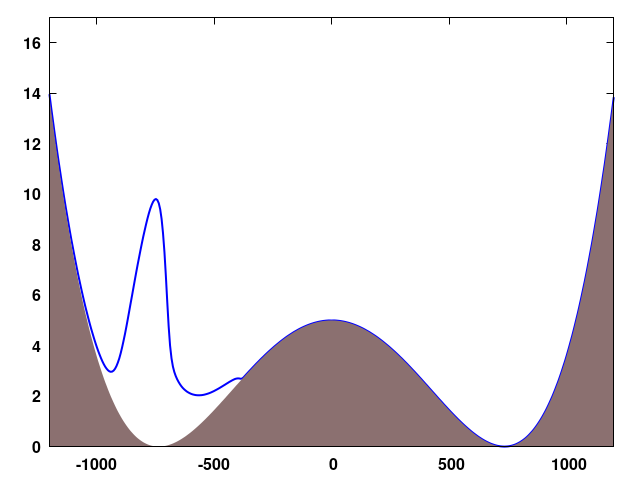}
       &\includegraphics[width=0.35\linewidth]{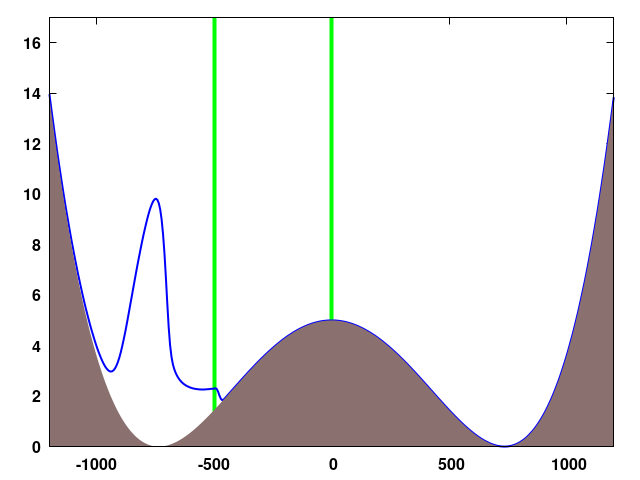}\\
    \begin{turn}{90}\hspace{10mm}\tiny{Water level
        [m]}\end{turn}
       &\includegraphics[width=0.35\linewidth]{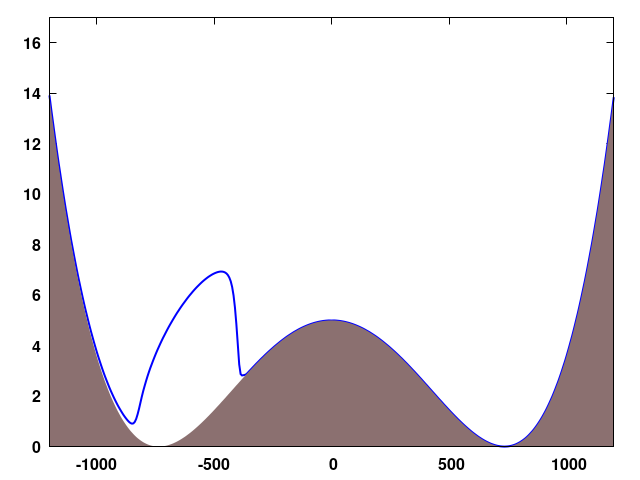}
       &\includegraphics[width=0.35\linewidth]{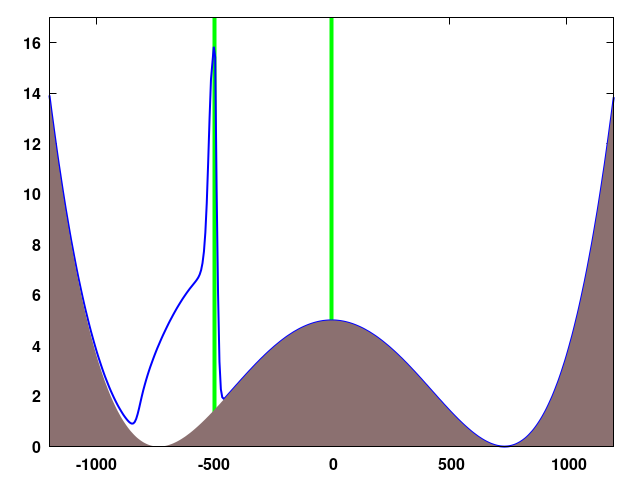}\\
    \begin{turn}{90}\hspace{10mm}\tiny{Water level
        [m]}\end{turn}
       &\includegraphics[width=0.35\linewidth]{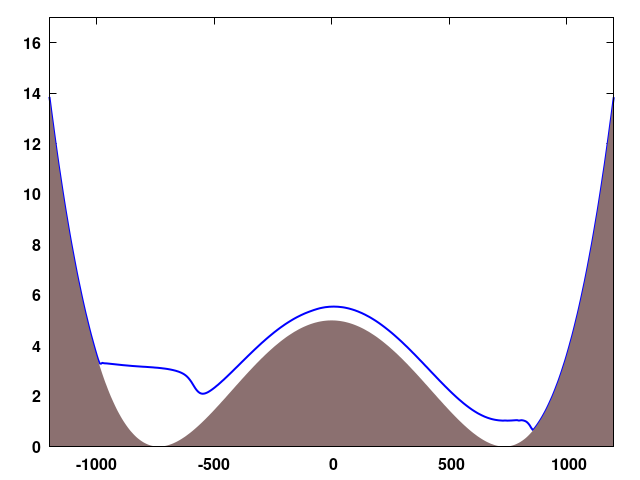}
       &\includegraphics[width=0.35\linewidth]{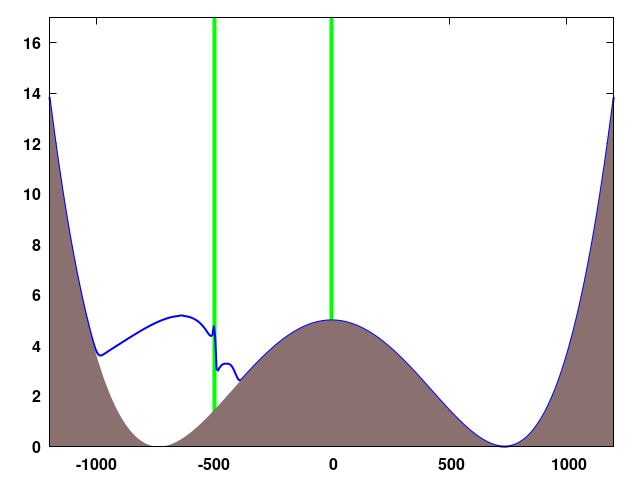}\\
    \begin{turn}{90}\hspace{10mm}\tiny{Water level
        [m]}\end{turn}
       &\includegraphics[width=0.35\linewidth]{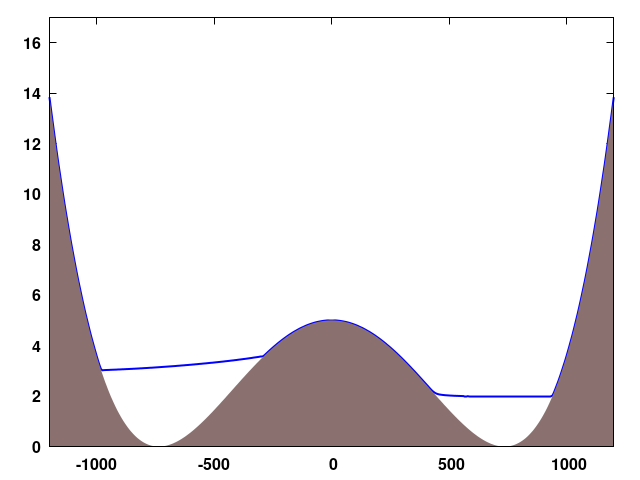}
       &\includegraphics[width=0.35\linewidth]{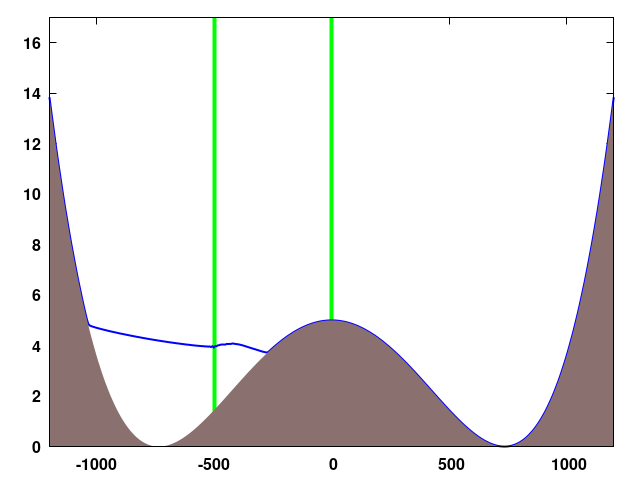}\\
    {} & \tiny{$x$ [m]} & \tiny{$x$ [m]}
  \end{tabular}
  \caption{Dumping effect of water flow.  The left column of
    images shows the states of the solution at four
    different moments of time ($t=40,\; 60,\; 200,\; 400$ s)
    for the configuration without vegetation.  The solution
    for the configuration with vegetation is represented in
    the right column at the same moments of time.  Remark
    the tendency of the water to form two lakes in the first
    configuration and only one lake in the second
    configuration.}
  \label{fig_1Dknoll}
\end{figure}
Figure~\ref{fig_1Dknoll} illustrates the numerical solution
of the water surface for these two configurations.

\subsubsection{Simulation on Paul's Valley}
\label{sect_SimulationonPaulValley}
Unfortunately, we do not have data for the water
distribution, plant cover density and measured velocity
field in a hydrographic basin to compare our numerical
results with.  However, to be closer to reality, we have
used GIS data for the soil surface of Paul's Valley and
accomplished a theoretical experiment: starting with a
uniform water depth on the entire basin and using different
cover plant densities, we run the 2D version of our scheme
(a flow modulus of ASTERIX) on a hexagonal network.  To
generate an arbitrary hexagonal network from an arbitrary
raster GIS data, we use the method introduced in
\cite{sds-ADataPortingTool}.
Figure~\ref{fig_2Dcomparison} shows that the numerical
results are consistent with direct observations concerning
the water time residence in the hydrographic basin.
\begin{landscape}
  \begin{figure}[ht]
    \centering
    \includegraphics[width=0.6\textwidth]{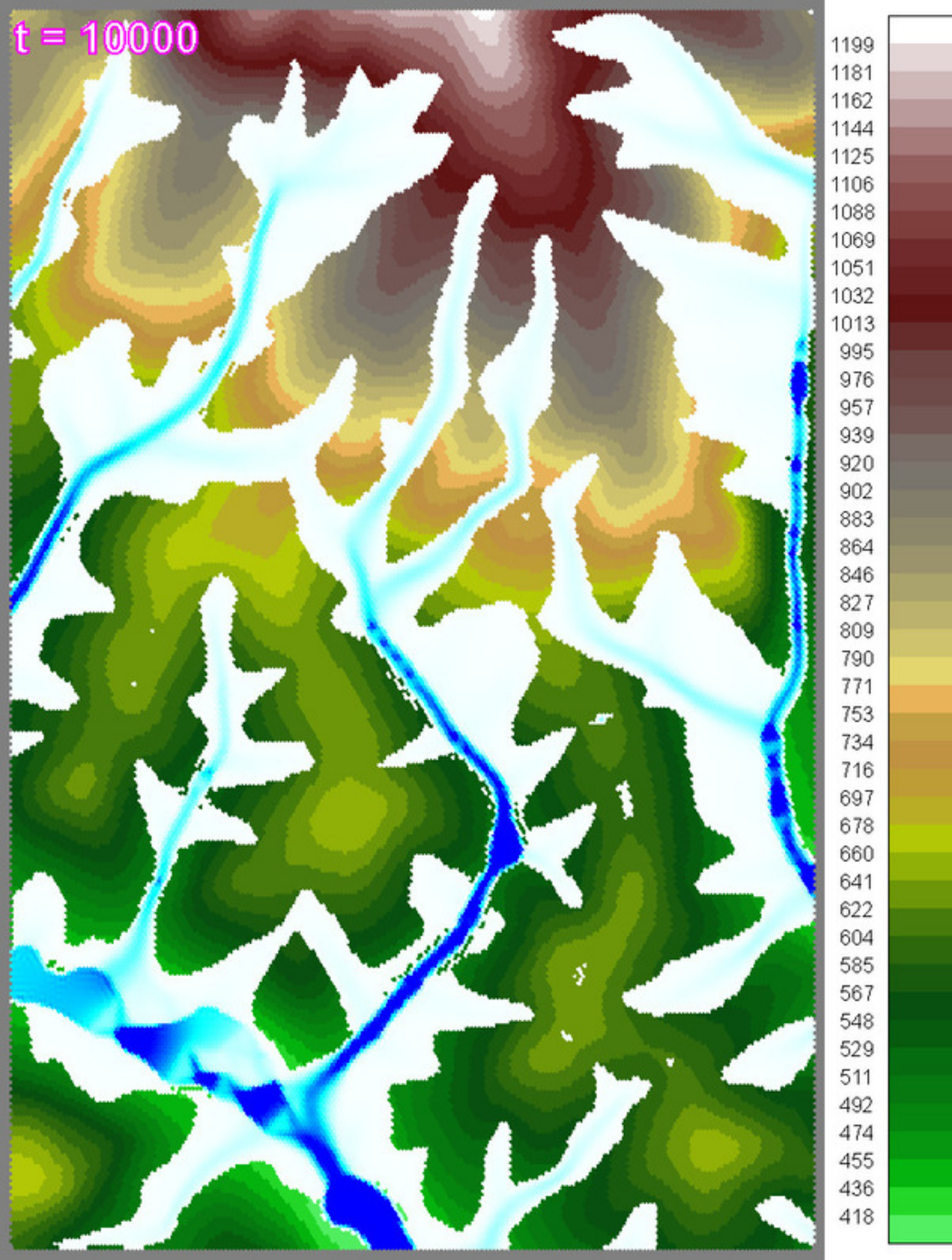}
    \hspace{1cm}
    \includegraphics[width=0.6\textwidth]{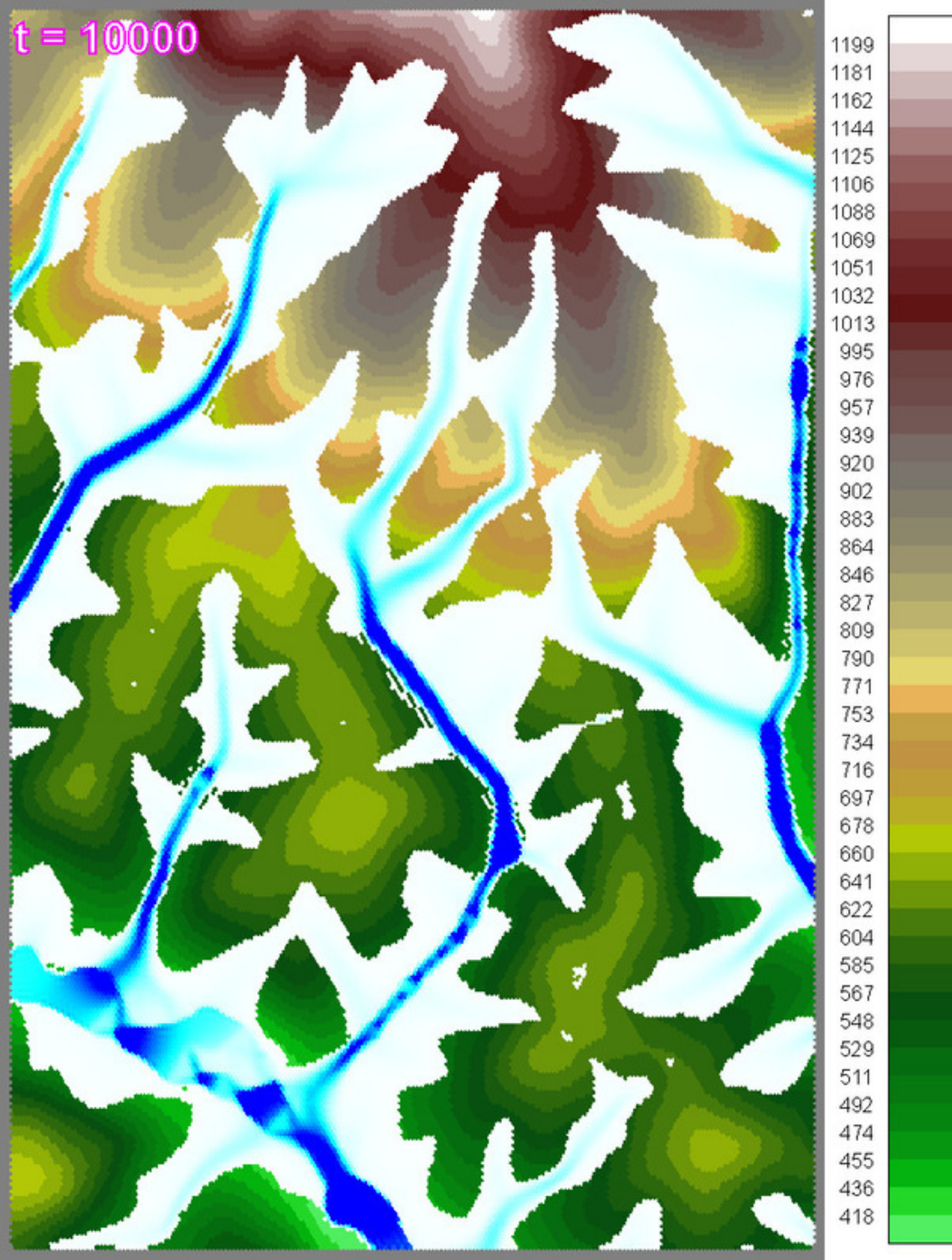}
    \caption{Snapshot of water distribution in Paul's Valley
      hydrographic basin.  Direct observations indicate that
      the water time residence depends on the density of the
      cover plant.  Our numerical data are consistent with
      terrain observations: the water drainage time is
      greater for the case of higher cover plant density.
      $\theta=0.97$ and $\theta=0.65$ for the left and right
      picture, respectively.}
    \label{fig_2Dcomparison}
  \end{figure}
\end{landscape}

\section{Discussion, Conclusions and Further Research}
The effect of vegetation on the water circulation in a
continuum soil-plant-atmosphere can relatively be easy
described, but it is hardly quantifiable.  The canopy of
vegetation intercept and store a certain amount of rain, the
stem of plant generate a resistance force to the water
running on the soil surface and the plant roots modify the
water infiltration rate into soil.  To quantify these
physical processes and put them in a mathematical model is a
big challenge and it requires new mathematical concept,
experimental results and numerical simulations.

The mathematical model (\ref{swe_vegm_rm.02}) is simplified
as possible as can be to retain essential interaction
factors and to be mathematical tractable.  The mathematical
model (\ref{swe_vegm_rm.02}) belongs to the class of
hyperbolic systems and it has certain distinctive feature
that raises special difficulties when one elaborates a
numerical scheme to approximate its solution.  Namely,

(a) the system becomes singular for $h=0$;

(b) the hyperbolic system has a non-conservative form and
its coefficients depend on the space variable,
non-homogeneous hyperbolic system; 

(c) there can be a dependence of the system of the
derivatives of discontinuous functions, soil surface and/or
porosity functions; 

(d) the unknown water depth function $h(t,x)$ must be a
positive function.

\medskip We appreciate that the finite volume method for the
space discretization of the system is the most fitted
method.  It allows one to treat the non-conservative
product, see (\ref{fvm_2D_eq.01-02}), and then perform the
approximations (\ref{fvm_2D_eq.03}-3), (\ref{fvm_2D_eq.05}),
(\ref{fvm_2D_eq.05_varianta2}) and
(\ref{fvm_2D_eq.05_varianta3}).  By introducing the flux
terms (\ref{fvm_2D_eq.04}-1), (\ref{fvm_2D_eq.04}-2),
(\ref{fvm_2D_eq.06}), we were able to ensure the positivity
of $h$.

To solve the problem of the singularities introduced by
$h=0$, we perform a first order fractional time step method.
In the presence of frictional force or drag resistance
force, the second equation in the scheme
(\ref{fvm_2D_eq_frac.06}) is nonsingular even if $h=0$.

When validating using Riemann problem, Section
\ref{sect_RiemannProblem}, we should observe that for
defining the standing waves, one needs to augment the
shallow water equations with an additional principle: either
path connection or an energetic principle.

For us, testing the numerical scheme against the solution of
Riemann problem, given by a composition of simple waves, is
instructive in two ways.  In one way, it is interesting to
find out how the numerical scheme works in all the possible
situations: non-existence, uniqueness and non-uniqueness of
the solutions.  On the other way, the numerical scheme does
not ``see'' the second principle invoked in the construction
of the standing waves since it is based on the approximation
principles of the operators and functions appearing in the
shallow water equations.  Thus, when one compares the
numerical solution to the analytic solution, the question of
compatibility between the two kind of principles can be
raised.  However, the case of the Godunov type approximation
schemes is different since these schemes are obtained on the
basis of the (exact or approximate) solution of the Riemann
problem.  We observed that it does not matter whether the
quantity $(\theta h)^s_{(i,j)}$ is calculated using
(\ref{fvm_2D_eq.05_varianta2}) or
(\ref{fvm_2D_eq.05_varianta3}) in the cases of non-existence
or uniqueness of the solution.  But in the case of multiple
solutions, different approximation formulas of the gradient
of the free surface lead to different solutions.  Also, we
draw attention that by non-existence case we mean "{\it the
  non-existence of a composite wave as a solution of the
  Riemann Problem}", which is different from "{\it the
  non-existence of a solution of the Riemann Problem}".
Consequently, the numerical solution can be an approximation
of the true solution or can be a fictitious mathematical
object.

Note that the transition from the non-existence to the
existence but not uniqueness of the solution for Riemann
problem is controlled only by the size of the jump in
$\theta$.

The numerical simulations for the steady
solution case of a flow over a bump (Section
\ref{sect_FlowOveraBump}) suggest the following conjecture.
\begin{conjecture*}
  There is a value $q^M$ such that for any $q\in (0, q^M)$
  there are two threshold values
  $0.2\leq\eta_1(q)<\eta_2(q)$ such that:
  \begin{enumerate}
  \item the numerical solution asymptotically reaches a
    state $(h^{\star},u^{\star})$;
  \item the state $(h^{\star},u^{\star})$ casts into one of
    the following three distinct categories:
    \begin{enumerate}
    \item for any $\eta\in(0.2,\eta_1)$, it develops a
      moving shock which goes out of the domain $\Omega$
      after a while,
    \item for any $\eta\in(\eta_1,\eta_2)$, it has a steady
      shock located in the domain occupied by the bump,
    \item for any $\eta>\eta_2$, it is smooth (no shock).
    \end{enumerate}
  \end{enumerate}
\end{conjecture*}
Case (a) of the conjecture is illustrated on the first row
of pictures in Figures~\ref{fig_validate_shm_shs_case1} and
\ref{fig_validate_shm_shs_case2}; case (b) is also
illustrated on the second row of pictures in the same
figures.

A common problem in many countries, Romania being a typical
example, is land deforestation with severe flood effects.
The test proposed in Section~\ref{sect_FlowOverConvex} was
considered to emphasize the protective role of vegetative
barrier against such natural disasters.  Remark that in the
configuration with vegetation, the water does not anymore
cross the knoll since it loses part of its energy due to the
friction with plants, see Figure~\ref{fig_1Dknoll}.  Note
that our numerical method is not only well-balanced, but it
is capable to produce solutions that reach different
equilibrium states (depending on the initial data and the
values of the parameters).  It should be also remarked that
the porosity $\theta$ plays the role a control function to
get different equilibrium points: one or two lakes.

Water flow in hydrographic basins is more complicated due to
variations in soil topography and in vegetation diversity.
In contrast to the 1D cases where we were able to easily
simulate the soil data, we now need 2D land and vegetation
data to be closer to reality.  Data for soil surface can be
obtained from topographic maps, GIS, LiDAR etc.  For the
purpose of modeling, one needs an interface to transfer data
between different scales.  In this sense, we used the tool
built and described in \cite{sds-ADataPortingTool}.
Simulation on Paul's Valley (with no rain and infiltration)
was considered in Section~\ref{sect_SimulationonPaulValley}
as a first step to get closer to the interesting cases from
ecology, biology, hydrology and terrestrial sciences.  An
important variable for such a hydrographic basin is the
amount of water leaving the basin.  This variable can be
measured using the relative amount $q$ of water in the basin
at the moment of time $t$ given by
\begin{equation*}
  q(t) = \displaystyle\frac{\displaystyle\int_{\Omega}h(t,x){\rm d}x}{\displaystyle\int_{\Omega}h(0,x){\rm d}x}.
\end{equation*}
\begin{figure}[htbp]
  \centering
  \includegraphics[width=0.49\textwidth]{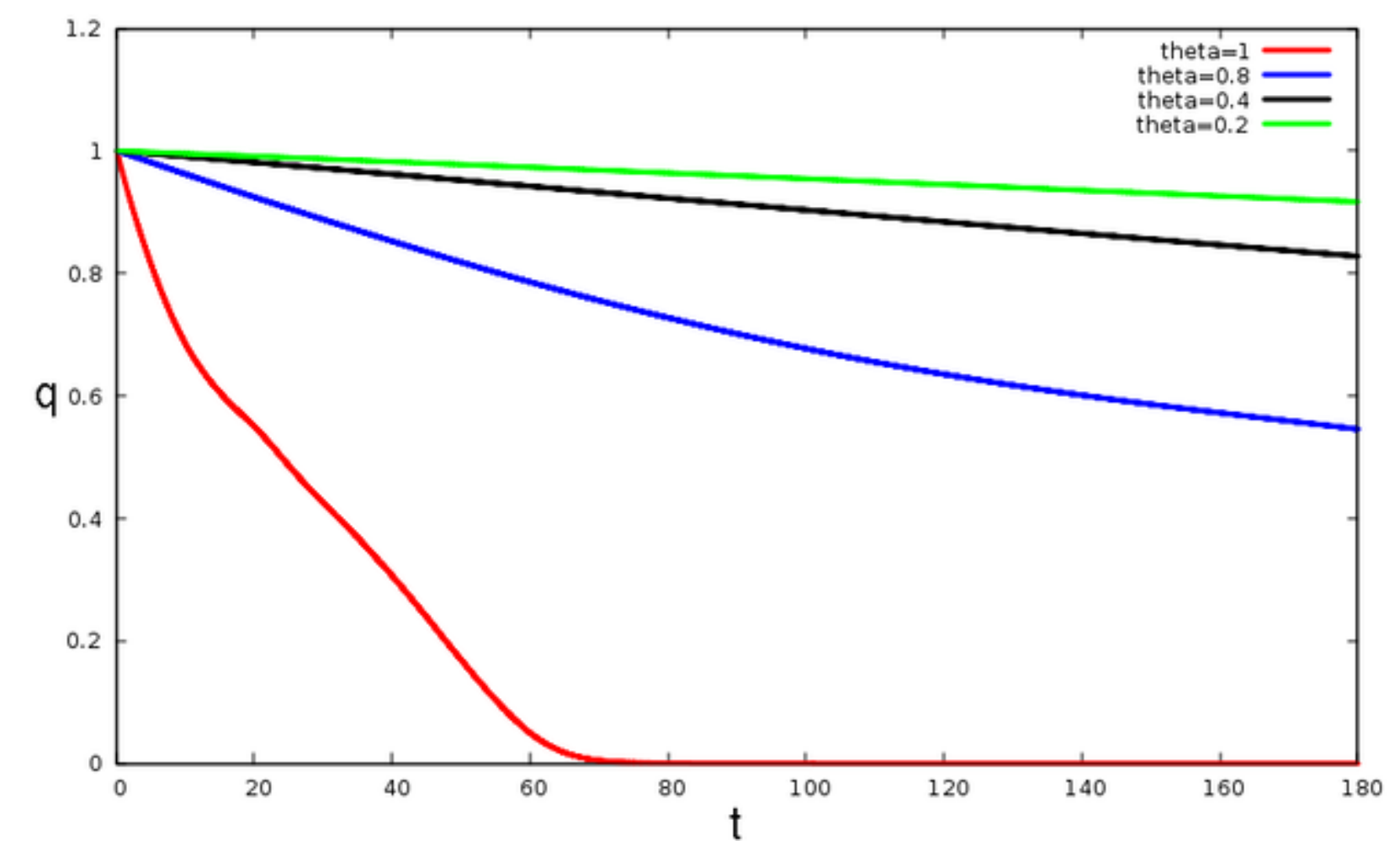}
  \includegraphics[width=0.49\textwidth]{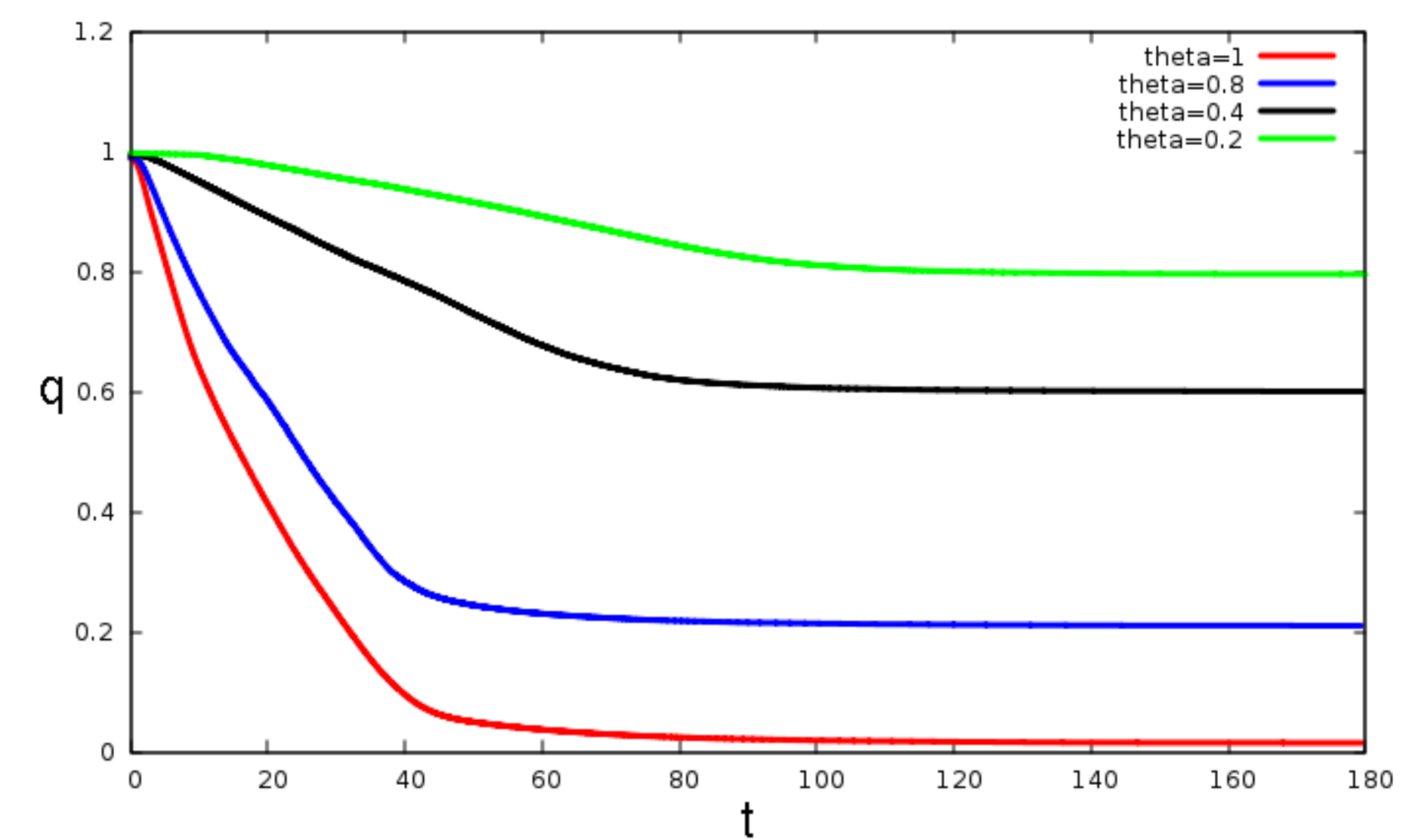}
  \caption{Time evolution of the water content in Paul's
    Valley hydrographic basin with ASTERIX (left picture)
    and CAESAR (right picture).}
  \label{fig_asterix_vs_caesar}
\end{figure}
Figure~\ref{fig_asterix_vs_caesar} shows the results for the
water content in Paul's Valley basin obtained with our
models ASTERIX and with CAESAR-Lisflood-OSE \cite{sds-ose}.

A general issue relates to whether higher cover plant
densities can prevent soil erosion and flood or not.  Both
pictures show that if the cover plant density is increasing
then the decreasing rate $\dot{q}$ of $q$ is smaller.  We
can think at a ``characteristic velocity'' of the water
movement in the basin and this velocity is in a direct
relation with $\dot{q}$.  We can now speculate that smaller
values of $\dot{q}$ imply softer erosion processes.

This valley belongs to Ampoi's catchment basin.  Flood
generally appears when the discharge capacity of a river is
overdue by the water coming from the river catchment area.
Our pictures show that higher cover plant densities imply
smaller values of $\dot{q}$ which in turn give Ampoi River
the time to evacuate the water amount flowing from the
valley.

As conclusions, we have presented a mathematical model able
to take into account the effects introduced by plants on the
water flow on soil surface and we have introduced a
numerical scheme to approximate its solution.  The numerical
scheme is of first order in time and space accurate and is
able to model a large class of physical phenomena, waves
propagation, backward wave reflection induced by obstacles
or filter band vegetation, oscillating flow, one or multiple
lake formation.  The scheme is relatively simple to be
implemented in a more complex model like soil erosion.

As current research, we work on elaborating a methodology
based on inverse method to estimate the hydrodynamics
parameters $\alpha_p$ and $\alpha_s$ from {\it in situ}
experimental measurements.

In the future, we intend to build up a coupled model --
shallow water - soil erosion -- to investigate the role of
plants to mitigate the transport of the polluted soil
particles from contaminated area to around area.

\section*{Acknowledgment}
This work was partially supported by a grant of the Ministry
of Research and Innovation, CCCDI-UEFISCDI, project number
PN-III-P1-1.2-PCCDI-2017-0721/34PCCDI/2018, within PNCDI
III.

\bibliographystyle{plain}
\bibliography{biblio_swmhpvh}

\end{document}